\documentclass[aps,pra,reprint,twocolumn,superscriptaddress,nofootinbib]{revtex4-1}
\pdfoutput=1 

\usepackage{graphicx}
\usepackage{bm}
\usepackage{amsmath}
\usepackage{amssymb}
\usepackage{amsfonts}
\usepackage{amsthm}
\usepackage{mathtools}
\usepackage{hyperref}
\usepackage{braket}
\usepackage[normalem]{ulem}
\usepackage{wrapfig}
\usepackage{tikz}
\usepackage{dsfont}
\usepackage{comment}
\usepackage{thmtools,thm-restate}
\mathtoolsset{showonlyrefs}

\hypersetup{colorlinks=true,urlcolor=[rgb]{0,0,0.5},citecolor=[rgb]{0.5,0,0},linkcolor=[rgb]{0,0,0.4}}

\usetikzlibrary{decorations.pathmorphing}

\newtheorem{theorem}{Theorem}
\newtheorem{definition}{Definition}
\newtheorem{corollary}{Corollary}
\newtheorem{lemma}{Lemma}

\newtheorem*{theorem*}{Theorem}
\newtheorem{task}{Task}
\newtheorem*{task*}{Task}
\newtheorem*{proposition*}{Proposition}

\def\autorefapp#1{\hyperref[#1]{Appendix~\ref{#1}}}

\def\tr{{\rm tr}}

\def\ketbra#1{ |{#1}\rangle\!\langle{#1}| }

\def\and{\quad {\rm and} \quad}

\newcommand{\calA}{\mathcal{A}}

\newcommand{\calC}{\mathcal{C}}
\newcommand{\calD}{\mathcal{D}}
\newcommand{\calE}{\mathcal{E}}

\newcommand{\calM}{\mathcal{M}}

\newcommand{\calS}{\mathcal{S}}
\newcommand{\calT}{\mathcal{T}}

\DeclarePairedDelimiter{\abs}{\lvert}{\rvert}
\DeclarePairedDelimiter{\norm}{\lVert}{\rVert}
\DeclarePairedDelimiter{\iprod}{\langle}{\rangle}
\DeclarePairedDelimiter{\brk}{[}{]}
\DeclarePairedDelimiter{\brc}{\{}{\}}

\makeatletter
\def\Pr{\@ifnextchar[{\@witha}{\@withouta}}
\def\@witha[#1]{\mathop{\operator@font Pr}_{#1}\brk}
\def\@withouta{\mathop{\operator@font Pr}\brk}
\makeatother

\makeatletter
\def\E{\@ifnextchar[{\@withb}{\@withoutb}}
\def\@withb[#1]{\mathop{\mathbb{E}}_{#1}\brk}
\def\@withoutb{\mathop{\mathbb{E}}\brk}
\makeatother

\makeatletter
\def\Var{\@ifnextchar[{\@withc}{\@withoutc}}
\def\@withc[#1]{\mathop{\mathbb{V}}_{#1}\brk}
\def\@withoutc{\mathop{\mathbb{V}}\brk}
\makeatother

\newcommand{\tvd}{d_{\text{TV}}}

\newcommand{\bone}{\mathds{1}\brk}

\renewcommand{\vec}[1]{\mathbf{#1}}
\newcommand{\T}{\vec{T}}
\newcommand{\U}{\vec{U}}
\newcommand{\V}{\vec{V}}
\newcommand{\W}{\vec{W}}
\newcommand{\Z}{\vec{Z}}
\newcommand{\Id}{\mathds{1}}
\newcommand{\Sig}{\vec{\Sigma}}

\newcommand{\rhomm}{\rho_{\mathsf{mm}}}

\newcommand{\R}{\mathbb{R}}

\newcommand{\KL}[2]{\text{KL}\left(#1\|#2\right)}

\newcommand{\Sodd}{S_{\mathsf{odd}}}
\newcommand{\Seven}{S_{\mathsf{even}}}

\newtheorem{fact}[theorem]{Fact}

\DeclareMathOperator{\Tr}{tr}
\DeclareMathOperator{\Wg}{Wg}

\begin{document}
\title{A Hierarchy for Replica Quantum Advantage}

\author{Sitan Chen}
\email{sitanc@berkeley.edu}
\affiliation{Department of Electrical Engineering and Computer Sciences, University of California, Berkeley, Berkeley, CA, USA}
\affiliation{Simons Institute for the Theory of Computing, Berkeley, CA, USA}

\author{Jordan Cotler}
\email{jcotler@fas.harvard.edu}
\affiliation{Society of Fellows, Harvard University, Cambridge, MA, USA}
\affiliation{Black Hole Initiative, Harvard University, Cambridge, MA, USA}

\author{Hsin-Yuan Huang}
\email{hsinyuan@caltech.edu}
\affiliation{
Institute for Quantum Information and Matter, Caltech, Pasadena, CA, USA}
\affiliation{Department of Computing and Mathematical Sciences, Caltech, Pasadena, CA, USA}

\author{Jerry Li}
\email{jerrl@microsoft.com}
\affiliation{Microsoft Research, Redmond, WA, USA}

\begin{abstract}
We prove that given the ability to make entangled measurements on at most $k$ replicas of an $n$-qubit state $\rho$ simultaneously, there is a property of $\rho$ which requires at least order $2^n$ measurements to learn.  However, the same property only requires one measurement to learn if we can make an entangled measurement over a number of replicas polynomial in $k, n$.  Because the above holds for each positive integer $k$, we obtain a hierarchy of tasks necessitating progressively more replicas to be performed efficiently.  We introduce a powerful proof technique to establish our results, and also use this to provide new bounds for testing the mixedness of a quantum state.
\end{abstract}

\maketitle

\section{Introduction}

In conventional physics experiments, an experimental sample (e.g. a superconductor, an array of atoms, a Bose-Einstein condensate, etc.) is prepared and then subsequently measured.  This is repeated numerous times to extract information about the system of interest.  But suppose an experimentalist has multiple \textit{replicas} of their experimental setup in the same lab -- could this dramatically increase the efficiency of learning about the system of interest?

In the simplest case, say that the experimentalist has two copies of their setup.  If the experimentalist can perform entangled measurements on both copies simultaneously, then remarkably, there are certain experiments which can be performed \textit{exponentially} more efficiently with two replicas versus with one~\cite{huang2021information, aharonov2021quantum, chen2021exponential, huang2021quantum, harrow2021approximate}.

There are several contemporary experimental platforms which can both (i) furnish several replicas, typically two to three; and (ii) allow for entangled measurements on said replicas~\cite{islam2015measuring, linke2018measuring, arute2019quantum}. 
Moreover, these platforms have recently been leveraged to experimentally measure purities and R\'{e}nyi entropies of novel quantum states~\cite{islam2015measuring, kaufman2016quantum, linke2018measuring, cotler2019quantum}. 
The most flexible experimental platforms of this kind are NISQ computers~\cite{preskill2018quantum, arute2019quantum}, which presently offer the widest range of possibilities for entangled measurements on replicas.

We are interested in the potential for a modest number of replicas to enable an exponential savings in resources for certain kinds of experiments; see the schematic in Figure~\ref{fig:replicas1}.  This could facilitate new forms of quantum experiments which are heretofore impractical or impossible.

A 2-versus-1 replica advantage can be achieved in a task about learning properties of an unknown quantum state. For instance, given an $n$-qubit experimental system $\rho$, the experimentalist wants to learn the expectation value of the absolute value of any observable that is a tensor product of Pauli operators.
The task requires $\sim \text{poly}(n)/\epsilon^2$ measurements if there are two replicas but $\sim \exp(n)/\epsilon^2$ measurements if there is only one, even if an adaptive strategy is employed~\cite{huang2021information, chen2021exponential}.
We say that this problem is \textit{1-replica hard}, since any strategy accessing one replica at a time is exponentially inefficient.  



While these results establish that there are problems which are $1$-replica hard and are rendered efficient with additional replicas, it is natural to ask if for each $k = 2,3,...$ there are problems which are $k$-replica hard and rendered efficient with additional replicas?  We provide such a family of $k$-replica hard problems.  Moreover, we show that each $k$-replica hard problem can be solved with $\text{poly}(k,n)$ replicas.  Our result provides a \textit{hierarchy} of problems which require more and more replicas to solve efficiently.



In the remainder of the paper, we outline the technical tools for our $\text{poly}(k,n)$-versus-$k$ replica advantage, and then state our main result.  Full details of the proofs are provided in the Appendix.  Then we present additional consequences of our results pertaining to mixedness testing.  We conclude with a discussion of implications and future directions.

\begin{figure*}[t!]
    \centering
    \includegraphics[width=0.9\textwidth]{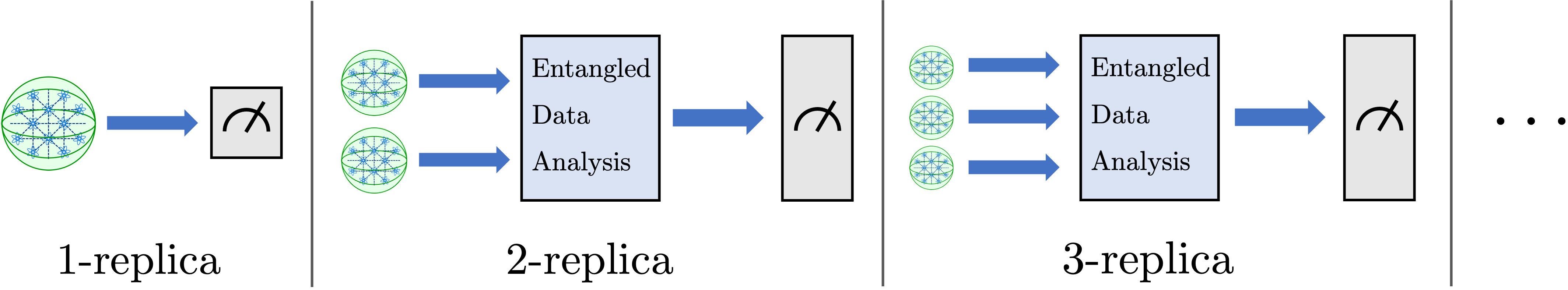}
    \centering
    \caption{\textit{Schematic of experiments with multiple replicas.}  Conventional experiments, including most contemporary experiments, fall into the 1-replica setting wherein only one copy of the system can be measured at a time.  In the multiple replica setting, several system copies can be mutually entangled via quantum computation (denoted by `entangled data analysis') and then jointly measured.}
    \label{fig:replicas1}
\end{figure*}

\section{Results on Replica Hierarchy}

We begin by formulating a family of learning problems which are $k$-replica hard.
Let $\mathcal{D}$ be a distribution over $n$-qubit unitaries formed by random quantum circuits with $\mathcal{O}(\mathrm{poly}(k, n))$ many gates.  We require that the distribution $\mathcal{D}$ forms a $\delta$-approximate unitary $2k$-design for $\delta=1/2^n$.  There are many random circuit architectures which are adequate for our purposes, e.g.~we can consider any geometrically local random quantum circuit of sufficient depth studied in \cite{brandao2016local, harrow2018approximate}.
We choose to consider a circuit architecture based on interleaving a $2\times 2$ unitary $V$ with random Clifford circuits \cite{haferkamp2020homeopathy};
see Appendix~\ref{sec:design} for more detailed definitions.
Assuming $n \gg k^2$, we consider the following task.
\begin{task}[RQC testing]\label{task1}
Determine whether an $n$-qubit state $\rho$ is the maximally mixed state $\mathds{1} / 2^n$ or the state $(\mathds{1} + \epsilon\, \U Z^{\otimes n} \U^\dagger) / 2^n$ where $\U$ is a fixed, random sample from the distribution $\mathcal{D}$ over random quantum circuits.
\end{task}
\noindent We prove the following result:
\begin{theorem}[$k$-replica hardness of the task]\label{thm:hardness}
The above task for $\epsilon = \frac{1}{3k}$ is $k$-replica hard.  In particular, any protocol which can make arbitrary entangled measurements of at most $k$ replicas of $\rho$ simultaneously must use $\Omega(2^n k)$ copies of $\rho$ overall to achieve the task with constant probability.
\end{theorem}

\noindent Here and throughout this work, we use $f(n) = \Omega(g(n))$ to denote that there is a constant $C > 0$ such that $f(n) \geq C\,g(n)$ for all sufficiently large $n$. A formal version of this Theorem is proved in the Appendix (see Theorem~\ref{thm:kdesign_formal}).
The theorem immediately implies that at least $\Omega(2^n)$ entangled measurements over $k$ replicas must be performed.


In contrast, if we are allowed to make a single entangled measurement over all replicas of $\rho$ simultaneously, then existing algorithms for shadow tomography \cite{aaronson2018shadow, aaronson2019gentle, buadescu2020improved} can solve the above task with exponentially fewer replicas.
In particular:
\begin{theorem}[$\text{poly}(k,n)$-replica easiness of the task]\label{thm:easiness} There is a protocol which performs a single entangled measurement on $\text{\rm poly}(k, n)$ replicas simultaneously and achieves the task with probability $9/10$.
\end{theorem}
\noindent We give a more detailed statement of this theorem along with its proof in Appendix~\ref{sec:upper}.




Let us provide intuition for Theorems~\ref{thm:hardness} and~\ref{thm:easiness}.  To understand Theorem~\ref{thm:hardness}, we need to elucidate why it is hard to distinguish a single batch of $k$ replicas of $\mathds{1} / 2^n$ from a single batch of $k$ replicas of $\frac{1}{2^n}(\mathds{1} + \tfrac{1}{3k}  \U Z^{\otimes n} \U^\dagger)$.
The idea is that a random quantum circuit $\U$ with sufficient depth leads to the state $\frac{1}{2^n}(\mathds{1} + \frac{1}{3k}\,\U Z^{\otimes n} \U^\dagger)$ having complicated correlations that cannot be resolved with a modest number of $k$-replica measurements.
Even given the ability to condition the choice of future $k$-replica measurements on past measurement outcomes, the correlations learned from \textit{any possible} $k$-replica measurement are so negligible that the adaptive strategy is futile unless there are exponentially many measurement rounds.
More technically, under any POVM measurement on $k$ replicas, a suitable statistical distance between $\frac{1}{2^n}\,\mathds{1}$ and $\frac{1}{2^n}(\mathds{1} + \frac{1}{3k} \, \U Z^{\otimes n} \U^\dagger)$ is exponentially small on average.

On the other hand, Theorem~\ref{thm:easiness} shows that when we can measure $\text{poly}(k,n)$ replicas in an entangled fashion, the desired task can be achieved with an economy of measurement.  The central technique is shadow tomography~\cite{aaronson2018shadow, aaronson2019gentle, buadescu2020improved}, which under certain conditions allows us to learn $M$ properties of a quantum state via an entangled measurement on $\text{polylog}(M)$ replicas of that state.  In our setting, $\mathcal{D}$ comprises only $\sim \exp(\text{poly}(k,n))$ distinct $\U$'s, and so we can construct as many distinct states $\frac{1}{2^n}(\mathds{1} + \frac{1}{3k} \, \U Z^{\otimes n} \U^\dagger)$.  Accordingly, we find that shadow tomography allows us to determine which particular state we have via a $\text{poly}(k,n)$-replica measurement.  Note that if $\mathcal{D}$ were to comprise \textit{general} random unitaries (i.e.~not formed by polynomial-depth circuits), there would instead be $\sim \exp(\exp(n))$ distinct unitaries, which is vastly more than $\sim \exp(\text{poly}(k,n))$.  Then shadow tomography would require an entangled measurement on $\sim \exp(n)$ replicas, as compared to the $\text{poly}(k,n)$ replicas in our setting.  In summary, we have deliberately parameterized our unitaries by polynomial-depth circuits so that shadow tomography can be performed with only $\text{poly}(k,n)$ replicas.

\section{Application to Mixedness Testing}

We can leverage our proof methods to provide new insights into mixedness testing.  The task of mixedness testing can be stated as follows:
\begin{task}[Mixedness Testing]\label{task2}
For an unknown density matrix $\rho$ on $\mathbb{C}^d$ and a fixed $\epsilon > 0$, determine whether $\rho$ is the maximally mixed state $\mathds{1}/d$ or instead $\| \rho - \mathds{1}/d\|_1 > \epsilon$.  (We are given a promise that one of these two scenarios is realized by $\rho$.)
\end{task}
\noindent That is, the task is asking us to determine if $\rho$ is the maximally mixed state or far from from it.  The hardness of this task clearly depends on the kinds of measurements we are allowed to perform, e.g.~if we are permitted to measure multiple copies of $\rho$ simultaneously.

In the context of experiments with replicas, it is natural to ask: if we have a $k$-replica system and thus can make joint, entangled measurements on at most $k$ copies of $\rho$ simultaneously, then how many copies of $\rho$ \textit{in total} do we require to solve the mixedness testing problem?  One might guess that it helps to have access to a larger number of replicas, so that the total number of $\rho$'s required will be smaller when $k$ is larger.

Interestingly, we can leverage our proof techniques used for Theorem~\ref{thm:hardness} to obtain new lower bounds for the problem of mixedness testing in the $k$-replica setting:

\begin{theorem}\label{thm:mixed}
    Suppose $k\in\mathbb{N}$ satisfies $k\le O(d^{1/2})$ and $k\ll \log(d)/\epsilon^2$. Any protocol which can make arbitrary entangled measurements of at most $k$ replicas of $\rho$ simultaneously must use at least $\Omega(d^{\frac{4}{3} - \eta}/\epsilon^2)$
    copies of $\rho$ overall, for any constant $\eta > 0$, to solve the mixedness testing task with constant probability.
\end{theorem}
\noindent This can be thought of as a strict strengthening of the main result in \cite{bubeck2020entanglement} which considered the special case of $k = 1$. In \cite{o2015quantum} it was shown that if one can make fully entangled measurements across all copies, then $\mathcal{O}(d/\epsilon^2)$ copies suffice for mixedness testing.  Therefore, we find that if we can perform measurements on at most $k$ replicas of $\rho$ simultaneously then for $k$ satisfying $k\le d^{1/2}$ and $k \ll \log(d)/\epsilon^2$, we require at least $\Omega(d^{4/3}/\epsilon^2)$ total copies of $\rho$ to solve the mixedness testing task whereas we only need $\mathcal{O}(d/\epsilon^2)$ copies if we can make entangled measurements on \textit{arbitrarily many} replicas.


\section{Discussion}

We have established a hierarchy of $\text{poly}(k,n)$-versus-$k$ replica advantages for learning about quantum states.  It would be nice to find variants of our hierarchy which are gate-efficient, in the sense that given $\text{poly}(k,n)$ replicas it is provably gate-efficient to distinguish $\mathds{1}/2^n$ from $\frac{1}{2^n}(\mathds{1} + \frac{1}{3k}\U Z^{\otimes n} \U^\dagger)$.  In our proof of Theorem~\ref{thm:easiness}, we leverage an existing algorithm for shadow tomography which is efficient in the number of replicas but is not necessarily gate-efficient.

Another important open question is to find a natural $f(k)$-versus-$k$ replica advantage for all $k$, where $f(k)$ is $n$-independent. More ambitiously, does there exist a $(k+1)$-versus-$k$ replica advantage?  This would tighten the replica quantum advantage hierarchy, and elucidate the added power of one additional replica.

Our investigations have been primarily limited to problems for learning about quantum states, but it seems quite likely that there exist hierarchies of replica advantages for learning about quantum channels as well.  Indeed, there are several known $2$-versus-$1$ replica advantages for learning properties of quantum channels~\cite{aharonov2021quantum, chen2021exponential}, and some hierarchy results regarding learning properties of elaborate quantum oracles encoding either a cryptographic generalization of Simon's problem~\cite{chia2020need} or the welded tree problem~\cite{coudron2020computations}. 
Note that in these problems, the quantum channels and the quantum oracles are not gate-efficient to implement.

A more general theoretical framework which goes beyond replica quantum advantage is quantum memory advantage.  In this context we can ask if there are tasks which can be efficiently carried out by a quantum computer with $\text{poly}(k,n)$ qubits of quantum memory versus $k \cdot n$ qubits of quantum memory. This allows for more complicated protocols than the replica setting.  For instance, in the state learning problems, suppose we have a $k\cdot n$ qubit quantum memory and we are storing $k$ $n$-qubit quantum states therein.  If we want to make room for an additional $n$-qubit state, then we could remove $n/k$ qubits from each of the $k$ $n$-qubit quantum states.
Such an operation is unlikely to be achieved in the near-term on more conventional experimental physics platforms since there are few means of interfacing experimental samples with an external quantum memory.  But this may be possible in the future.
While a gate-efficient quantum memory advantage is proven in~\cite{huang2021quantum}, establishing a \textit{hierarchy} of quantum memory advantage is an interesting open question.

Finally, it would be very interesting to experimentally implement replica quantum advantage in the NISQ era~\cite{preskill2018quantum}.
We anticipate that there may be interesting \textit{polynomial} replica quantum advantages that would have utility in realistic experimental settings, via the aid of entangled experimental apparatuses that can jointly measure several system copies.
Generalizing conventional experimental platforms to furnish this capability appears to be a worthy goal for future technological development.

\vspace*{6pt}
\noindent {\bf Acknowledgments.}\quad 
We thank Aram Harrow, Misha Lukin, Ryan O'Donnell, and John Wright for valuable discussions. SC is supported by the National Science Foundation under Award 2103300. JC is supported by a Junior Fellowship from the Harvard Society of Fellows, the Black Hole Initiative, as well as in part by the Department of Energy under grant {DE}-{SC0007870}.  HYH is supported by the Google PhD Fellowship. Part of this work was completed while SC and JL were visiting the Simons Institute for the Theory of Computing.

\pagebreak
\onecolumngrid
\appendix

\section{Technical Preliminaries}

\subsection{Notation}

Given an even $d\in\mathbb{N}$, let $\rhomm\coloneqq \frac{1}{d}\Id$ denote the maximally mixed state on $d$ qudits, and let $\Z\coloneqq \text(1,\ldots,-1,\ldots)$ denote the diagonal matrix with $d/2$ $1's$ and $d/2$ $-1$'s. When $d = 2^n$, we will use $\Z$ interchangeably with $Z^{\otimes n}$, where $Z$ denotes the $2\times 2$ Pauli $Z$ matrix.  Given $\U\in U(d)$, define $\rho_{\U}(\epsilon) \coloneqq \frac{1}{d}(\Id + \epsilon \U\Z\U^{\dagger})$. When $\epsilon$ is clear from context, we will refer to this as $\rho_{\U}$.

Given $S\subseteq[k]$ and matrices $A,B\in\mathbb{C}^{d\times d}$, we will let $A\otimes_S B$ denote the $k$-fold tensor whose $i$-th component is $A$ if $i\in S$ and $B$ otherwise.

Given a matrix $M\in\mathbb{C}^{d\times d}$, we will let $\norm{M}_p$ denote its Schatten-$p$ norm. Likewise, given a scalar random variable $Z$, we will denote $\norm{Z}_p \coloneqq \E{|Z|^p}^{1/p}$. When $p = \infty$, we will denote $\norm{M}_{\infty}$ by $\norm{M}$. When $p = 2$, we will denote $\norm{M}_2$ by $\norm{M}_{HS}$. 

Given functions $f,g:\Omega\to\R$ and a distribution $\mu$ over $\Omega$, we will use the notation
\begin{equation}
    \iprod{f,g}_{\mu} \coloneqq \E[v\sim\mu]{f(v)g(v)}. \label{eq:iprod}
\end{equation}

Let $\calS_{m}$ denote the symmetric group on $m$ elements. Given permutation $\pi\in\calS_{m}$, we let $\#\pi$ denote the number of cycles.

\subsection{Unitary Designs}
\label{sec:design}
Let us begin by defining approximate unitary designs.
\begin{definition}\label{def:kdesign}
    Given a distribution $\calD$ over $U(d)$, define the $2k$-fold twirling channel $\calC^{(k)}_{\calD}$ by 
    \begin{equation}
        \calC^{(2k)}_{\calD}[\rho] \coloneqq \E[\U\sim\calD]*{\U^{\otimes 2k}\rho {\U^{\dagger}}^{\otimes 2k}}
    \end{equation}
    for all $d$-qudit mixed states $\rho$. Denote by $\calC^{(2k)}_{\text{\rm Haar}}$ the twirling channel when $\calD$ is the Haar measure.
    
    In general, we say that $\calD$ forms a \emph{$\delta$-approximate $2k$-design} if
    \begin{equation}
        \norm{\calC^{(2k)}_{\calD} - \calC^{(2k)}_{\text{\rm Haar}}}_{\diamond} \le \frac{\delta}{d^{2k}},
    \end{equation} where $\norm{\cdot}$ denotes the diamond norm.
\end{definition}

We will use the following result stating that certain random quantum circuits are approximate designs. We first define these circuits:

\begin{definition}[Random $V$-Interleaved Clifford Circuits]\label{def:interleaved}
    Fix a $2\times 2$ unitary matrix $V$ and consider the following distribution $\calD$ over $2^n\times 2^n$ unitary matrices. To sample once from $\calD$, sample $D$ independent draws $\xi_1,\ldots,\xi_D$ from the uniform distribution over the set $\brc{V\otimes \Id_{2}^{\otimes (n-1)}, V^{\dagger}\otimes \Id_2^{\otimes (n-1)}, \Id_2^{\otimes n}}$, sample $D$ independent draws $\chi_1,\ldots,\chi_D$ from the uniform distribution over the Clifford group on $n$ qubits, and output the product $\chi_D \xi_D \cdots \chi_2 \xi_2 \chi_1 \xi_1$. We call a sample from $\calD$ a \emph{random $V$-interleaved Clifford circuit on $n$ qubits of depth $D$}.
\end{definition}

\noindent A schematic of a $V$-interleaved Clifford circuit on $n = 8$ qubits of depth $D = 4$ is shown in Figure~\ref{fig:interleaved}.

\begin{theorem}[Theorem 1 from \cite{haferkamp2020homeopathy}]\label{thm:homeopathy}
    Let $V\in U(2)$ be any non-Clifford unitary matrix, and let $d = 2^n$ for some $n\in\mathbb{N}$. There are absolute constants $c_1,c_2 > 0$ depending on $V$ such that for any depth $D \ge c_1\log^2(k)(k^4 + k\log(1/\delta))$ and any $n \ge c_2 k^2$, the ensemble of random $V$-interleaved Clifford circuits on $n$ qubits forms a $\delta$-approximate $k$-design. 
\end{theorem}


\begin{figure}[t]
    \centering
    \includegraphics[width=4cm]{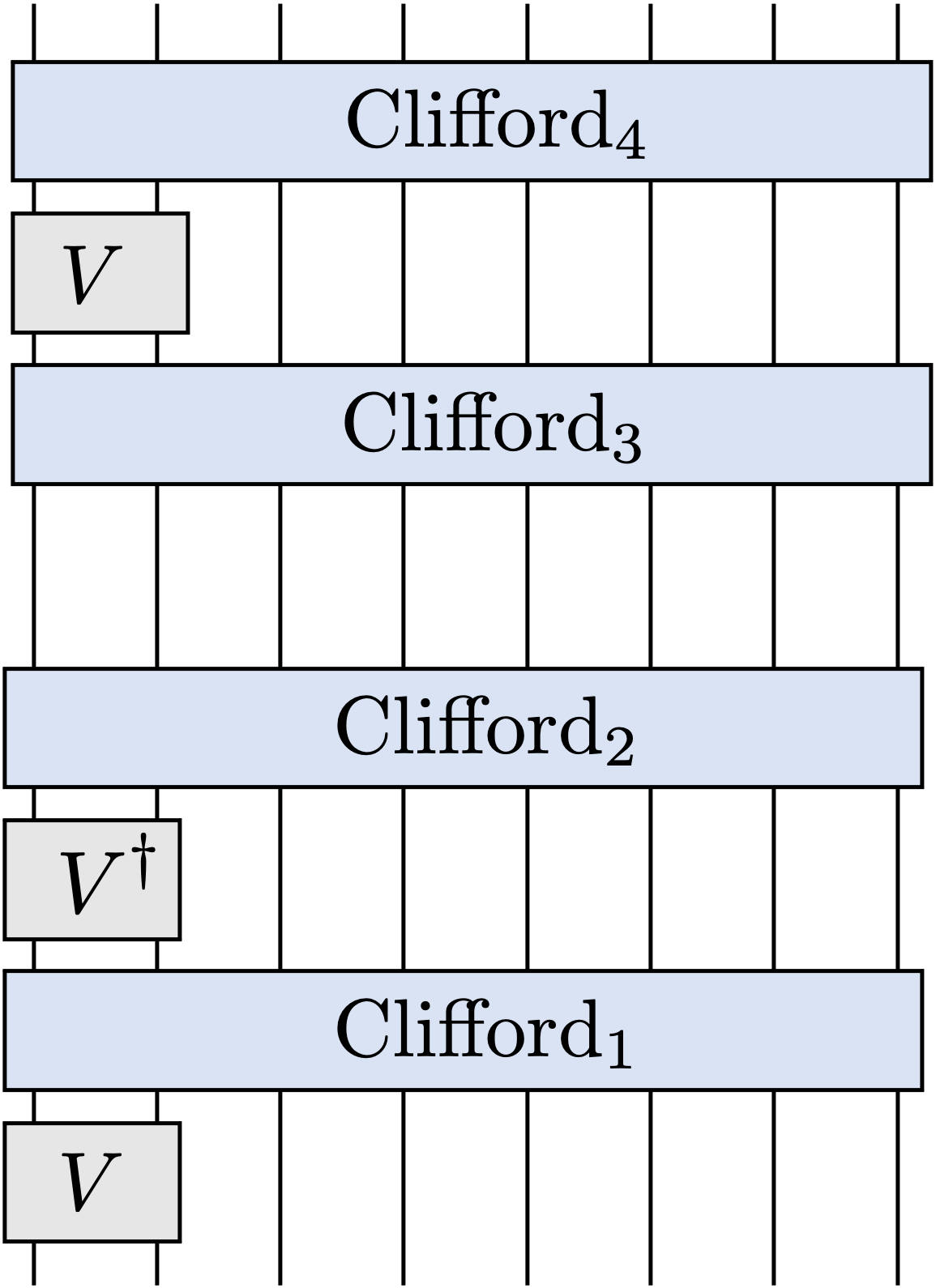}
    \centering
    \caption{\textit{Diagram of a $V$-interleaved circuit.} Shown above is a $V$-interleaved circuit on $8$ qubits with depth $4$.  We have $\xi_1 = V \otimes \Id_2^{\otimes (n-1)}$, $\xi_2 = V^{\dagger}\otimes \Id_2^{\otimes (n-1)}$, $\xi_3 = \Id_2^{\otimes n}$, $\xi_4 = V^\otimes \Id_2^{\otimes (n-1)}$, and $\chi_i = \text{Clifford}_i$. }
    \label{fig:interleaved}
\end{figure}

\subsection{Miscellaneous Combinatorial Facts}

\begin{fact}\label{fact:sumcycles}
    For any $m\in\mathbb{N}$, $\sum_{\pi\in \calS_m} d^{\#\pi} = d(d+1)\cdots(d+m-1)$.
\end{fact}

\begin{fact}\label{fact:evencycles}
    For any even $m$, $\sum_{\tau\in \calS_m\,\text{\rm even}} d^{\#\tau} = (m - 1)!!\cdot d(d+2)\cdots (d+m-2)$, where the sum is over $\tau$ which consist only of even cycles. For any odd $m$, this quantity is zero.
\end{fact}


\subsection{Weingarten Calculus}

Consider the Haar measure over the $d \times d$ unitary group $U(d)$.  The $m$th moment of $U(d)$ with respect to the Haar measure can be written as
\begin{equation}
\label{E:HaarUint1}
\mathbb{E}_{U \sim \text{Haar}}\!\left[U_{i_1 j_1} U_{i_2 j_2} \cdots U_{i_m j_m} U_{i_1' j_1'}^\dagger U_{i_2' j_2'}^\dagger \cdots U_{i_m' j_m'}^\dagger\right] = \sum_{\sigma, \tau \in\calS_m} \delta_{\sigma(I), J'} \delta_{\tau(J), I'} \text{Wg}(\sigma \tau^{-1}, d)
\end{equation}
where on the right-hand side we have used the multi-index notation $I = (i_1, i_2,...,i_m)$, as well as the notation
\begin{equation}
\delta_{\sigma(I), J'} := \delta_{i_{\sigma(1)}, j_1'} \delta_{i_{\sigma(2)}, j_2'} \cdots \delta_{i_{\sigma(m)}, j_m'}\,.
\end{equation}
The function $\text{Wg}(\,\cdot\,, d) :\calS_m \to \mathbb{R}$ is called the Weingarten function.  Note that~\eqref{E:HaarUint1} implies the formula
\begin{equation}
\label{E:AUBUformula}
\mathbb{E}_{U \sim \text{Haar}}\!\left[ \tr{(AUBU^\dagger)}^m\right] = \sum_{\sigma, \tau \in\calS_m} \text{tr}(\tau^{-1} A^{\otimes m}) \text{tr}(\sigma B^{\otimes m}) \, \text{Wg}(\sigma \tau^{-1},d)
\end{equation}
which we will make use of later.

We record the following estimates of Weingarten functions which will be utilized below.
\begin{lemma}[Lemma 6 from \cite{aharonov2021quantum}]\label{lem:weinsum}
$\sum_{\tau \in\calS_m} |\text{\rm Wg}(\tau, d)| = \frac{(d-m)!}{d!}$\,.
\end{lemma}

\begin{lemma}[Lemma 16 from \cite{montanaro2013weak}]\label{lem:montanaro}
For $\ell \le d^{2/3}$ and $\pi\in\calS_{m}$, $\Wg(\pi,d) \le O(d^{\#\pi - 2 m})$.
\end{lemma}

We will also make crucial use of concentration for the Haar measure over the unitary group:

\begin{theorem}\label{thm:levy}
    For any $F:U(d)\to\R$ which is $L$-Lipschitz, the random variable $F(\U)$ for Haar-random $\U\in U(d)$ is $\sigma^2$-sub-Gaussian for $\sigma = O(L/\sqrt{d})$.
\end{theorem}

\section{Upper Bound}
\label{sec:upper}

In this section we show that an algorithm with the ability to measure $\sim n^9 k^{12}$ replicas simultaneously can solve the distinguishing task in Theorem~\ref{thm:kdesign_formal} with exponentially fewer replicas than any algorithm which can only measure $k$ replicas at a time:

\begin{theorem}\label{thm:upperbound}
    There is an algorithm that can solve the distinguishing task in Theorem~\ref{thm:kdesign_formal} with probability at least $9/10$ using at most $\widetilde{O}(n^6 k^{12} (n^3 + \log^3(1/\epsilon)))$ copies of $\rho$ by making a joint measurement on all copies.
\end{theorem}
\noindent We will use the algorithm of \cite{buadescu2020improved} for ``quantum hypothesis selection.'' Here we record a special case of their guarantee:

\begin{theorem}[Special Case of Theorem 1.5 from \cite{buadescu2020improved}]\label{thm:tomography_algo}
    There is an algorithm that, given a classical description of $m$ mixed states $\sigma_1,\ldots,\sigma_m$, parameters $0 < \epsilon < 1/2$, takes $N$ copies of $\sigma_i$ for an unknown $i\in[m]$,
    for
    \begin{equation}
        N = O\left(\frac{\gamma}{\epsilon^2}\cdot (\log^3 m + \gamma \log m)\right),
    \end{equation} where $\gamma\coloneqq \log\log(1/\epsilon)$, performs some fully entangled measurement on the copies of $\rho$, and with probability at least $9/10$ outputs $i'\in[m]$ satisfying $\norm{\rho - \sigma_{i'}} < \epsilon$.
\end{theorem}

We can now prove Theorem~\ref{thm:upperbound}:

\begin{proof}[Proof of Theorem~\ref{thm:upperbound}]
    First note that because the Clifford group on $n$ qubits has size $\exp(O(n^2))$, the support of $\calD$ has size \begin{equation}
        m = 3^D\cdot \exp(O(n^2 D)) = \exp(O(n^2 D)),
    \end{equation}
    where in $D = c_1\log^2(k) (k^4 + kn + k\log(1/\epsilon))$ for some absolute constant depending on $V$.
    Now consider the following algorithm. Given \begin{align}
        N &= O\left(\frac{\gamma}{\epsilon^2}\cdot (\log^3 m + \gamma \log m)\right) \\
        &= O\left(\frac{\gamma}{\epsilon^2} (n^6 D^3 + \gamma n^2 D)\right) \\
        &\le \widetilde{O}(n^6 k^{12} (n^3 + \log^3(1/\epsilon)))
    \end{align} copies of $\rho$, where $\gamma = \log\log(1/\epsilon)$, run the algorithm of Theorem~\ref{thm:tomography_algo} for the collection of mixed states $\sigma_1,\ldots,\sigma_m$ given by $\rhomm$ and all $\rho_{\U}$ for $\U$ in the support of $\calD$, taking the accuracy parameter $\epsilon$ in Theorem~\ref{thm:tomography_algo} to be $\epsilon$ in Theorem~\ref{thm:upperbound}. The guarantee of that algorithm is that with probability at least $9/10$, the algorithm succeeds and we get an index $i'\in[m]$ for which $\norm{\rho - \rho_{i'}}_1 < \epsilon$. Our algorithm will decide that $\rho = \rhomm$ if $i'$ is the index of $\rhomm$ in the collection $\brc{\sigma_j}$, and otherwise it will decide that $\rho$ came from the ensemble of states $\rho_{\U}$.
    
    It remains to justify that our algorithm successfully solves the distinguishing task with probability at least $9/10$. First, if $\rho = \rhomm$, note that $\norm{\rho - \rho_{\U}} = \epsilon$, so if the algorithm succeeds, $\rho_{i'}$ has to be $\rhomm$. By the same reasoning, if $\rho = \rho_{\U}$ for some $\U$ in the support of $\calD$, then it cannot be that $\rho_{i'} = \rhomm$. This concludes the proof.
\end{proof}

\section{Lower Bound Framework}
\label{sec:framework}

In this section we overview our general framework for proving sample complexity lower bounds, which follows that of~\cite{chen2021exponential}. We will be interested in \emph{distinguishing tasks}: given the ability to make some number $N$ of measurements of an unknown state $\rho$ which is promised to lie in one of two possible sets $S_0$ and $S_1$, we must decide based on the \emph{transcript}, i.e. the sequence of measurement outcomes we observe, whether the object lies in $S_0$ or $S_1$.

We now formalize what it means for an algorithm to solve a distinguishing task by making entangled measurements of $k$ replicas of $\rho$ at a time:

\begin{definition}[Tree representation for testing algorithms,~\cite{chen2021exponential}] \label{def:treelearnstate}
    Fix an unknown quantum state $\rho\in\mathbb{C}^{d\times d}$. A testing algorithm making entangled measurements on $k$ replicas of $\rho$ at a time can be expressed as a pair $(\mathcal{T},\mathcal{A})$. $\mathcal{T}$ is a rooted tree of depth $N$, where each node in the tree encodes all the measurement outcomes the algorithm has seen so far. The tree satisfies the following properties:
    \begin{itemize}
        \item Each node $u$ is associated with a probability $p_{\rho}(u)$.
        \item For the root $r$ of the tree, $p_{\rho}(r) = 1$.
        \item At each non-leaf node $u$, we measure a POVM $\{\mathcal{M}^u_v\}_{v}$ on $\rho^{\otimes k}$ to obtain a classical outcome $v$.
        \item If $v$ is a child node of $u$, then
                \begin{equation}
                    p_{\rho}(v) = p_{\rho}(u) w^u_v \cdot d^k \bra{\psi^u_v} \rho^{\otimes k} \ket{\psi^u_v},
                \end{equation}
                where $\{w^u_v \cdot d^k \ketbra{\psi^u_v}\}_v$ is a rank-$1$ POVM that depends on the node $u$.
        \item Every root-to-leaf path is of length $N$. 
    \end{itemize}
    For any node $u$ at depth $t$, $p_{\rho}(u)$ is the probability that the transcript is given by the sequence of outcomes in the path from $r$ to $u$ after $t$ measurements. $\mathcal{A}$ is any classical algorithm that takes as input a transcript corresponding to any leaf node $\ell$ and outputs $i\in\brc{0,1}$
\end{definition}

Observe that any tree $\mathcal{T}$ induces a probability distribution over transcripts of length $N$, and to any $\rho$ is associated a distribution $p^{\le N}_{\rho}$ over transcripts. When $N$ is clear from context, we will omit the superscript.

Formally we must argue that if $N$ is insufficiently large, then no algorithm $\calA$ mapping transcripts $T$ to $\brc{0,1}$ can satisfy $\Pr[T\sim p_{\rho}]{\calA(T) = i} \ge 2/3$ for all $\rho\in S_i$ and $i\in\brc{0,1}$ (the constant 2/3 could be replaced with an arbitrary constant greater than 1/2).

A standard trick to show this is to reduce to the following \emph{average-case} question. For $i\in\brc{0,1}$, let $\calD_i$ be some distribution over $S_1$ that we are free to choose. We want to show that for $N$ insufficiently large, it is information-theoretically impossible to distinguish between the following two scenarios: 1) $\rho$ is sampled from $\calD_0$, and we observe a transcript drawn from $p_{\rho}$, and 2) likewise but $\rho$ is sampled from $\calD_1$. The following fundamental result in binary hypothesis testing shows that to establish a lower bound for the original distinguishing task, it suffices to establish a lower bound for this average-case task.


\begin{lemma}[Le Cam's two-point method]\label{lem:lecam}
    If there existed distributions $\calD_0, \calD_1$ over $S_0, S_1$ for which
    \begin{equation}
    d_{\text{\rm TV}}\left(\E[\rho\sim\calD_0]{p_\rho}, \E[\rho\sim\calD_1]{p_\rho}\right) < 1/3\,,    
    \end{equation}
    then no algorithm $\calA$ mapping transcripts of length $N$ to $\brc{0,1}$ can satisfy $\Pr[T\sim p_{\rho}]{\calA(T) = i} \ge 2/3$ for all $\rho\in S_i$ and $i\in\brc{0,1}$.
\end{lemma}

\begin{proof}
    Suppose to the contrary that there existed such an algorithm. For $i\in\brc{0,1}$, let $q_i \coloneqq \E[\rho\sim\calD_i]{p_\rho}$. Then 
    \begin{equation}
        2/3 \le \E[\rho\sim\calD_i]*{\Pr[T\sim p_\rho]{\calA(T) = i}} = \Pr[T\sim q_1]{\calA(T) = i}\,.
    \end{equation} But by Fact~\ref{fact:hypo_test} below, this would contradict the fact that $\tvd(q_0,q_1) < 1/3$.
\end{proof}

The proof of Lemma~\ref{lem:lecam} uses the following basic fact:

\begin{fact}\label{fact:hypo_test}
    Given distributions $q_0, q_1$ over a domain $S$, if $\tvd(q_0,q_1) < 1/3$, there is no algorithm $\calA: S\to\brc{0,1}$ for which $\Pr[x\sim q_i]{\calA(x) = i} \ge 2/3$ for both $i = 0,1$.
\end{fact}

\begin{proof}
    Let $S'\subseteq S$ denote the set of elements $x$ for which $\calA(x) = 0$. Then observe that
    \begin{align}
        \Pr[x\sim q_0]{\calA(x) = 1} + \Pr[x\sim q_1]{\calA(x) = 0} &= 1 - q_0(S') + q_1(S') \\
        &\ge 1 - \sup_{S''\subseteq S}\abs*{q_0(S'') - q_1(S'')} \\
        &= 1 - \tvd(q_0,q_1) \ge 2/3,
    \end{align} so at least one of the terms on the left-hand side is at least $1/3$.
\end{proof}

In the sequel, the choice of $\calD_i$ will be fairly immediate, and the bulk of the work will go into bounding the total variation distance between $\E[\rho\sim\calD_0]{p_\rho}$ and $\E[\rho\sim\calD_i]{p_\rho}$. The following basic fact will be useful to this end, as it says that to upper bound the total variation between two distributions, it suffices to give a pointwise upper bound on their likelihood ratio:

\begin{fact}\label{fact:oneside}
    Let $0 \le \delta < 1$. Given probability distributions $p, q$ on a finite domain $\Omega$, if $p(x)/q(x) > 1 - \delta$ for all $x\in\Omega$, then $\tvd(p,q) \le \delta$.
\end{fact}

\section{Proof of Hardness for RQC Testing (Theorem~\ref{thm:hardness})}
\label{sec:proof_kdesign}

In this section, we prove the following formal version of Theorem~\ref{thm:hardness}:

\begin{theorem}\label{thm:kdesign_formal}
    Fix any non-Clifford unitary $V\in U(2)$. There are absolute constants $c, c_1, c_2 > 0$ such that for any $\epsilon\in[0,1]$ and $k,n\in\mathbb{Z}$ satisfying $\epsilon\le 1/3k$, $k \le c 2^{n/2}/\epsilon$, and $n \ge c_2 k^2$, the following holds. 
    
    Consider the task of distinguishing between whether an unknown state $\rho$ is the maximally mixed state or the state $\frac{1}{2^n}(\Id_{2^n} + \epsilon \U Z^{\otimes n} \U^{\dagger})$ where $\U$ is a random $V$-interleaved Clifford circuit on $n$ qubits (see Definition~\ref{def:interleaved}) of depth at least $D = c_1\log^2(k)(k^4 + kn + k\log(1/\epsilon))$.
    
    Any algorithm which can make arbitrary entangled measurements of $k$ replicas of $\rho$ at a time must use $\Omega(2^n/(k\epsilon^2))$ replicas overall to solve this distinguishing task with nontrivial advantage.
\end{theorem}

For convenience, define $d\coloneqq 2^n$. The only structure we will use about the distribution over random Clifford circuits $\U$ is the fact that it forms a $\delta$-approximate $2k$-design. In the framework of the previous section, we will take the two sets $S_0$ and $S_1$ to be $S_0 = \brc{\rhomm}$ and $S_1 = \brc{\rho_{\U}(\epsilon)}_{\U}$ where $\U$ ranges over unitaries in the support of a $\delta$-approximate $2k$-design. The choice of $\calD_0,\calD_1$ is clear: as $S_0$ is a singleton set, $\calD_0$ is just the point mass at $\rhomm$, whereas $\calD_1$ is the distribution over $S_1$ given by the approximate $2k$-design.

First, we show how to upper bound the total variation distance between $p_{\rhomm}$ and $\E[\U\sim\calD]{p_{\rho_{\U}(\epsilon)}}$ in terms of the quantity $\E[\U\sim\calD]{(\delta^{\U}_j)^2}$ for a worst-case choice of POVM $\brc{F_j}$. Henceforth, for notational convenience we will refer to $\rho_{\U}(\epsilon)$ as $\rho_{\U}$.

\begin{lemma}\label{lem:convexity}
    For $\epsilon \le 1/3k$, define
    \begin{equation}
        \Phi^{\U}_0(\epsilon) \coloneqq \sum_{\substack{S\subset[k]:\\ |S|\, \text{\rm even},\, S\neq\emptyset}}\epsilon^{|S|} \cdot \U\Z\U^{\dagger}\otimes_S \Id \qquad \Phi^{\U}_1(\epsilon) \coloneqq \sum_{\substack{S\subset[k]:\\ |S|\,\text{\rm odd}}} \epsilon^{|S|}  \cdot \U\Z\U^{\dagger}\otimes_S \Id.
    \end{equation}
    Then we have that
    \begin{equation}
        d_{\text{\rm TV}}\left(p_{\rhomm},\E{p_{\rho_{\U}}}\right) \le 2N\max_{\ket{\psi}} \brc*{\E*{\bra{\psi}\Phi^{\U}_1(\epsilon)\ket{\psi}^2} + 2\E*{\bra{\psi}\Phi^{\U}_0(\epsilon)\ket{\psi}^2}^{1/2}}\label{eq:tvbound}
    \end{equation}
    where the expectations are with respect to an arbitrary distribution $\calD$ over $U(d)$.
\end{lemma}

\begin{proof}
    Let $T$ be any transcript of length $N$, and consider the root-to-leaf path $\brc{u_0,\ldots,u_n}$ in $\calT$ associated to this transcript. If the POVM used at node $u$ along this path is given by $\brc{w^u_v \cdot d^k\cdot \ketbra{\psi^u_v}}_v$, then
    \begin{align}
        \E[\U]*{\frac{p_{\rho_{\U}}(T)}{p_{\rhomm}(T)}} &= \E[\U]*{\prod^N_{t = 1} \bra{\psi^{u_{t-1}}_{u_t}}(\Id + \U\Z\U^{\dagger})^{\otimes k}\ket{\psi^{u_{t-1}}_{u_t}}}  \\
        &= \E[\U]*{\prod^N_{t = 1}\left(1 + \bra{\psi^{u_{t-1}}_{u_t}}\left((\Id + \U\Z\U^{\dagger})^{\otimes k} - \Id^{\otimes k}\right)\ket{\psi^{u_{t-1}}_{u_t}}\right)} \\
        &=: \E[\U]*{\prod^N_{t=1} \left(1 + \delta^{\U}_t\right)} \,. \label{eq:prod}
    \intertext{where $\delta^{\U}_t \coloneqq \bra{\psi^{u_{t-1}}_{u_t}}\left((\Id + \U\Z\U^{\dagger})^{\otimes k} - \Id^{\otimes k}\right)\ket{\psi^{u_{t-1}}_{u_t}}$. By concavity of the exponential function, we can lower bound \eqref{eq:prod} by}
        &\ge \exp\left(\sum^N_{t=1}\E[\U]*{\log\left(1 + \delta^{\U}_t\right)}\right).  \label{eq:explog}
    \intertext{Note that the Haar measure on $U(d)$ is invariant under left-multiplication by $\T = \begin{pmatrix}
        \vec{0} & \Id \\
        \Id & \vec{0}
    \end{pmatrix}$, so}
        &\ge \exp\left(\frac{1}{2}\sum^N_{t = 1} \E[\U]*{\log\left(1 + \delta^{\U}_t\right) + \log\left(1 + \delta^{\T\U}_t\right)}\right) \\
        &\ge \exp\left(\frac{1}{2}\sum^N_{t = 1} \E[\U]*{\log\left(\left(1 + \delta^{\U}_t\right)\left(1 + \delta^{\T\U}_t\right)\right)}\right) \\
        &\coloneqq \exp\left(\frac{1}{2}\sum^N_{t=1}\E[\U]*{\log\left(1 - \gamma^{\U}_t\right)}\right) \qquad \left(\gamma^{\U}_t \coloneqq 1 - (1 + \delta^{\U}_t)(1 + \delta^{\T\U}_t)\right)\\
        &\ge \exp\left(\frac{1}{2}\sum^N_{t=1}\E[\U]*{\log\left(1 - \max(0,\gamma^{\U}_t)\right)}\right) \\
        &\ge \exp\left(\frac{1}{2}\sum^N_{t=1}\E[\U]*{-4\max(0,\gamma^{\U}_t)}\right) \label{eq:lateruseful}\\
        &\ge 1 - 2\sum^N_{t=1}\E[\U]*{\max(0,\gamma^{\U}_t)}\,,\label{eq:useTU}
    \end{align}
    where in the penultimate step we used the fact that $\max(0,\gamma^{\U}_t) \le 0.96$ for every $t$ because
    \begin{equation}
        |\delta^{\U}_t| \le \abs*{\sum_{S\neq\emptyset} \epsilon^{|S|}\bra{\psi^{u_{t-1}}_{u_t}}\U\Z\U^{\dagger}\otimes_S \Id\ket{\psi^{u_{t-1}}_{u_t}}} \le (1 + \epsilon)^k - 1 \le 0.4
    \end{equation} by the assumption that $\epsilon \le 1/3k$.

    Now note that
    \begin{align}
        \gamma^{\U}_t &= 1 - \left(1 + \delta^{\U}_t\right)\left(1 + \delta^{\T\U}_t\right) \\
        &= 1 - \sum_{S,S'\subseteq[k]}(-1)^{|S'|}\epsilon^{|S|+|S'|}\bra{\psi^{u_{t-1}}_{u_t}}\U\Z\U^{\dagger}\otimes_S \Id\ket{\psi^{u_{t-1}}_{u_t}}\bra{\psi^{u_{t-1}}_{u_t}}\U\Z\U^{\dagger}\otimes_{S'} \Id\ket{\psi^{u_{t-1}}_{u_t}}  \\
        &= 1 - \bra{\psi^{u_{t-1}}_{u_t}}(\Id + \Phi^{\U}_0(\epsilon))\ket{\psi^{u_{t-1}}_{u_t}}^2 + \bra{\psi^{u_{t-1}}_{u_t}}\Phi^{\U}_1(\epsilon)\ket{\psi^{u_{t-1}}_{u_t}}^2\\
        &= \bra{\psi^{u_{t-1}}_{u_t}}\Phi^{\U}_1(\epsilon)\ket{\psi^{u_{t-1}}_{u_t}}^2 - \bra{\psi^{u_{t-1}}_{u_t}}\Phi^{\U}_0(\epsilon)\ket{\psi^{u_{t-1}}_{u_t}}^2 - 2\bra{\psi^{u_{t-1}}_{u_t}}\Phi^{\U}_0(\epsilon)\ket{\psi^{u_{t-1}}_{u_t}}\,,
    \end{align}
    so
    \begin{equation}
        \E[\U]*{\max(0,\gamma^{\U}_t)} \le \max_{\ket{\psi}} \brc*{\E[\U]*{\bra{\psi}\Phi^{\U}_1(\epsilon)\ket{\psi}^2} + 2\E[\U]*{\abs*{\bra{\psi}\Phi^{\U}_0(\epsilon)\ket{\psi}}}}\,, \label{eq:relubound}
    \end{equation}
    where the maximum is over unit vectors $\ket{\psi}$. Substituting this into \eqref{eq:useTU} and applying Cauchy-Schwarz and Fact~\ref{fact:oneside} yields the lemma.
\end{proof}

Next, we provide estimates for the quantities on the right-hand side of \eqref{eq:tvbound} in the case where $\calD$ is an \emph{exact} $2k$-design.

\begin{lemma}\label{lem:secondmoment}
    For $k \le O(\min(d,\sqrt{d}/\epsilon))$ and any unit vector $\ket{\psi}$, if $\calD$ is a $2k$-design, then we have
    \begin{enumerate}
        \item $\E[\U\sim\calD]*{\bra{\psi}\Phi^{\U}_1(\epsilon)\ket{\psi}^2} \le O(\epsilon^2k^2/d)$
        \item $\E[\U\sim\calD]*{\bra{\psi}\Phi^{\U}_0(\epsilon)\ket{\psi}^2} \le O(\epsilon^4k^4/d^2)$
    \end{enumerate}
\end{lemma}

\begin{proof}
    Let $\Sodd$ and $\Seven$ denote the set of nonempty subsets of $[k]$ of odd and even cardinality respectively. Given $S\subset[k]$, let $\rho_S\in(\mathbb{C}^{d})^{\otimes|S|}$ denote $\ketbra{\psi}$ with the components indexed by $[k]\backslash S$ traced out so that
    \begin{equation}
        \bra{\psi}\U\Z\U^{\dagger}\otimes_S\Id\ket{\psi} = \Tr\left((\U^{\dagger})^{\otimes|S|}\rho_S \U^{\otimes|S|}\Z^{\otimes|S|}\right) \label{eq:rhoS}
    \end{equation}
    \noindent{\textbf{Proof of 1:}} Using \eqref{eq:rhoS}, we can write
    \begin{equation}
        \E{\bra{\psi}\Phi^{\U}_1(\epsilon)\ket{\psi}^2} = \sum_{S,S'\in\Sodd} \epsilon^{|S| + |S'|} \E*{\Tr\left((\U^{\dagger})^{\otimes|S|}\rho_S \U^{\otimes|S|}\Z^{\otimes|S|}\right)\Tr\left((\U^{\dagger})^{\otimes|{S'}|}\rho_{S'} \U^{\otimes|{S'}|}\Z^{\otimes|{S'}|}\right)}. \label{eq:summands_odd}
    \end{equation}
    For any fixed $S,S'\in\Sodd$, if $m \coloneqq |S| + |S'|$, then the expectation of the corresponding summand in \eqref{eq:summands_odd} is
    \begin{equation}
        \epsilon^m\sum_{\pi,\tau\in \calS_m} \Wg(\pi\tau^{-1},d) \Tr(\pi(\rho_S\otimes \rho_{S'})) \Tr(\tau Z^{\otimes m}). \label{eq:weingarten_kdesign}
    \end{equation} We can na\"{i}vely upper bound $\abs*{\Tr(\pi(\rho_S\otimes \rho_{S'}))}$ by 1. Recall by Lemma~\ref{lem:weinsum} that $\sum_{\pi\in \calS_m}|\!\Wg(\pi,d)| = \frac{(d-m)!}{d!} \le 1/\Omega(d)^m$ for $m \le d$ and also note that $\Tr(\tau Z^{\otimes m})$ is zero unless $\tau$ has only even cycles, in which case $\Tr(\tau Z^{\otimes m}) = d^{\#\tau}$. By Fact~\ref{fact:evencycles}, $\sum_{\tau\in\calS_m\,\text{even}} d^{\#\tau} = (m-1)!!\cdot d(d+1)\cdots (d+m-2) \le m^{m/2}\cdot O(d)^{m/2}$. We conclude that \eqref{eq:weingarten_kdesign} is upper bounded by $O(\epsilon^2 m/d)^{m/2}$. Substituting this bound into \eqref{eq:summands_odd}, we get an upper bound of
    \begin{align}
        \sum_{S,S'\in\Sodd} O((|S| + |S'|)/d)^{|S|/2+|S|'/2} &= \sum_{2\le m \le k \ \text{even}} O(\epsilon^2 m/d)^{m/2} \cdot \sum_{0<i<m \ \text{odd}} \binom{k}{i}\binom{k}{m - i} \\
        &\le \sum_{2\le m \le k \ \text{even}} O(\epsilon^2 m/d)^{m/2}\cdot \binom{2k}{m} \\
        &\le \sum_{2\le m \le k \ \text{even}} O(\epsilon^2 m/d)^{m/2} \cdot O(k/m)^m \\
        &\le \sum_{2\le m \le k \ \text{even}} O\left(\frac{\epsilon^2 k^2}{md}\right)^{m/2} \\
        &\le \sum_{2\le m \le k \ \text{even}} O\left(\frac{\epsilon^2 k^2}{d}\right)^{m/2}.
    \end{align} If $k\le c\sqrt{d}/\epsilon$ for sufficiently small absolute constant $c > 0$, then this geometric series is dominated by its first summand, yielding the desired bound of $O(\epsilon^2 k^2/d)$.
\\ \\
    \noindent{\textbf{Proof of 2:}} Proceeding in the same way as above, we can upper bound $\E{\bra{\psi}\Phi^{\U}_1(\epsilon)\ket{\psi}^2}$ by
    \begin{align}
        \sum_{S,S'\in\Seven} O((|S| + |S'|)/d)^{|S|/2+|S|'/2} &= \sum_{4\le m \le k \ \text{even}} O(\epsilon^2 m/d)^{m/2} \cdot \sum_{0<i<m \ \text{even}} \binom{k}{i}\binom{k}{m - i} \\
        &\le \sum_{4\le m \le k \ \text{even}} O\left(\frac{\epsilon^2 k^2}{d}\right)^{m/2},
    \end{align}
    so again, if $k \le c\sqrt{d}/\epsilon$ then this geometric series is dominated by its first summand, yielding the desired bound of $O(\epsilon^4 k^4/d^2)$.
\end{proof}

Next we extend the estimates of Lemma~\ref{lem:secondmoment} to handle approximate $2k$-designs:

\begin{lemma}\label{lem:secondmoment_approx}
    For $k \le O(\min(d,\sqrt{d}/\epsilon))$ and any unit vector $\ket{\psi}$, if $\calD$ is a $\delta$-approximate $2k$-design, then we have
    \begin{enumerate}
        \item $\E[\U\sim\calD]*{\bra{\psi}\Phi^{\U}_1(\epsilon)\ket{\psi}^2} \le O(\epsilon^2k^2/d + \delta)$
        \item $\E[\U\sim\calD]*{\bra{\psi}\Phi^{\U}_0(\epsilon)\ket{\psi}^2} \le O(\epsilon^4k^4/d^2 + \delta)$
    \end{enumerate}
\end{lemma}

\begin{proof}
    Let $\calD^*$ denote the Haar measure over $U(d)$. It suffices to upper bound \begin{equation}
        \E[\U\sim\calD]*{\bra{\psi}\Phi^{\U}_i(\epsilon)\ket{\psi}^2} - \E[\U\sim\calD^*]*{\bra{\psi}\Phi^{\U}_i(\epsilon)\ket{\psi}^2} \label{eq:diff}
    \end{equation}
    in absolute value for $i = 0,1$ and then apply Lemma~\ref{lem:secondmoment} and triangle inequality.
    
    First note that for any $S\subseteq[k]$, we can write
    \begin{equation}
        \Tr\left({\U^{\dagger}}^{\otimes |S|} \rho_S \U^{\otimes |S|} \Z^{\otimes |S|}\right) = \Tr\left({\U^{\dagger}}^{\otimes k} (\rho_S\otimes \rhomm^{\otimes k-|S|}) \U^{\otimes k} (\Z^{\otimes |S|}\otimes \Id^{\otimes k - |S|})\right).
    \end{equation}
    For $i = 0$, \eqref{eq:diff} can thus be written as
    \begin{equation}
        \sum_{S,S'\in \Seven} \epsilon^{|S| + |S'|} \Tr\left(\left(\rho_S\otimes \rho_{S'}\otimes \rhomm^{2k - |S| - |S'|}\right)\left(\calC^{(2k)}_{\calD} - \calC^{(2k)}_{\text{Haar}}\right)\left[\Z^{\otimes |S|} \otimes \Z^{\otimes |S'|} \otimes \Id^{2k - |S| - |S'|}\right] \right).\label{eq:diff2}
    \end{equation}
    As $\calD$ is a $\delta$-approximate $2k$-design,
    \begin{equation}
        \norm*{\left(\calC^{(2k)}_{\calD} - \calC^{(2k)}_{\text{Haar}}\right)\left[\Z^{\otimes |S|} \otimes \Z^{\otimes |S'|} \otimes \Id^{2k - |S| - |S'|}\right]}_1 \le \frac{\delta}{d^{2k}}\cdot \norm{\Z^{\otimes |S|} \otimes \Z^{\otimes |S'|} \otimes \Id^{2k - |S| - |S'|}}_1 = \delta
    \end{equation}
    On the other hand, $\norm{\rho_S \otimes \rho_{S'}\otimes \rhomm^{2k - |S| - |S'|}} \le d^{-2k +|S|+|S'|}$, so we may upper bound \eqref{eq:diff2} in absolute value by
    \begin{equation}
        \frac{\delta}{d^{2k}}\sum_{S,S'\in\Seven} \epsilon^{|S| + S'|}\cdot d^{|S| +|S'|} = \frac{\delta}{d^{2k}} \left(\sum_{1\le m \le k \ \text{odd}} (\epsilon d)^m\cdot \binom{k}{m} \right)^2 \le \delta (1/d + \epsilon)^{2k} \le O(\delta),
    \end{equation}
    where the last step follows by the fact that $k \le d$ and $\epsilon < 1$.
    As the upper bound in the second step follows from na\"{i}vely including all summands, odd and even, we immediately also get an upper bound for \eqref{eq:diff} when $i = 1$.
\end{proof}

We are now ready to prove Theorem~\ref{thm:kdesign_formal}:

\begin{proof}[Proof of Theorem~\ref{thm:kdesign_formal}]
    Let $d = 2^n$ and $\delta = \epsilon^2/d$. By Theorem~\ref{thm:homeopathy}, the distribution over states $\rho_{\U}(\epsilon)$ for $\U$ given by the random Clifford circuits in Theorem~\ref{thm:kdesign_formal} is a $\delta$-approximate $2k$-design. Call this distribution $\calD$. By applying both parts of Lemma~\ref{lem:secondmoment} to the bound in Lemma~\ref{lem:convexity}, we get an upper bound on the total variation distance between $p_{\rhomm}$ and $\E[\U\sim\calD]{p_{\rho_{\U}(\epsilon)}}$ of $O(N\cdot(\epsilon^2 k^2/d + \delta))$. So for $N = O\left(\frac{d}{k^2\epsilon^2 + \delta d}\right) = O(d/(k^2\epsilon^2))$ sufficiently small (corresponding to $Nk = O(d/(k\epsilon^2))$ copies of the state in total), the total variation distance is less than $1/3$ and we can apply Lemma~\ref{lem:lecam}.
\end{proof}

\section{RQC Testing and Shadow Tomography}
\label{subsec:tomo}

Here we record the simple observation that hardness for Task~\ref{task1} implies hardness for shadow tomography with the collection of observables given by $\brc{\U Z^{\otimes n}\U^{\dagger}}_{\U\in\text{supp}(\calD)}$.  Note that $\text{supp}(\calD)$ is a finite set, by virtue of Definition~\ref{def:interleaved} and the fact that the Clifford group is discrete.  First, we recall the objective of shadow tomography:

\begin{task}[Shadow tomography]
    Given a classical description of $m$ observables $O_1,\ldots,O_m$ satisfying $\norm{O_i} \le 1$ for all $i$, and given access to copies of $\rho$, estimate each $\Tr(O_i\rho)$ to within $\epsilon$ additive error.
\end{task}

\begin{theorem}\label{thm:shdaow_kdesign}
    Fix any non-Clifford unitary $V\in U(2)$. There are absolute constants $c,c_1,c_2 > 0$ such that for any $\epsilon\in[0,1]$ and $k,n\in\mathbb{Z}$ satisfying $\epsilon \le 1/3k$, $k\le c2^{n/2}/\epsilon$, and $n \ge c_2 k^2$, the following holds.
    
    Let $O_1,\ldots,O_m$ consist of all observables of the form $\U Z^{\otimes n}\U^{\dagger}$ for $\U$ given by a $V$-interleaved Clifford circuit on $n$ qubits with depth $c_1\log^2(k)(k^4 + kn + k\log(1/\epsilon))$. Shadow tomography to error $\epsilon$ with respect to this collection of observables is $k$-replica hard. In particular, any protocol which can make arbitrary entangled measurements of at most $k$ replicas of $\rho$ simultaneously must use $\Omega(2^n / (k\epsilon^2))$ copies of $\rho$ overall to achieve the task with constant probability.
\end{theorem}

\begin{proof}
    Note that for $\rho = \rhomm$, $\Tr(O_i \rho) = 0$ for all $i$, whereas for $\rho = \rho_{\U}(2\epsilon)$, the observable $O_i$ corresponding to $\U Z^{\otimes n}\U^{\dagger}$ satisfies $\Tr(O_i \rho) = 2\epsilon$. So an algorithm solving shadow tomography to error better than $\epsilon$ can be used to distinguish between whether $\rho = \rhomm$ or whether $\rho = (\Id +\epsilon\,\U Z^{\otimes n}\U^{\dagger})/2^n$ for random circuit $\U$. The result then follows from hardness of the latter task, given by Theorem~\ref{thm:kdesign_formal}.
\end{proof}

\section{Proof of Mixedness Testing Separation (Theorem~\ref{thm:mixed})}

While the overall architecture of the proof will follow that of \cite{bubeck2020entanglement,chen2021toward}, the $k$-entangled setting introduces key subtleties that we will highlight in the sequel.

\subsection{Additional Notation}
\label{subsec:morenotation}

As in Section~\ref{sec:proof_kdesign}, for notational convenience we will refer to $\rho_{\U}(\epsilon)$ as $\rho_{\U}$.

Unlike in the previous section, a uniform lower bound on the likelihood ratio between $\E{p_{\rho_{\U}}(\epsilon)}$ and $p_{\rhomm}$ will not suffice here, and we need to be more careful in considering the contribution of different root-to-leaf paths in the tree $\mathcal{T}$. For a node $u$, a child $v$, let \begin{equation}
    p_0(v|u) \coloneqq w^u_v\,, \qquad p_{\U}(v|u) \coloneqq w^u_v  \cdot d \bra{\psi^u_v}{\rho_{\U}^{\otimes k}}\ket{\psi^u_v} \label{eq:probs}
\end{equation} denote the probability that the learner reaches node $v$ conditioned on reaching node $u$, provided the learner is measuring copies of $\rhomm$ or copies of $\rho_{\U}$ respectively. We will also use $p_0(\cdot|u)$ and $p_{\U}(\cdot|u)$ to refer to the corresponding conditional probability distributions. Additionally, it will be convenient to define $p^t_0(\cdot)$ to be the distribution over nodes at depth $t$ of $\calT$ induced by the first $t$ measurements when the unknown state is $\rhomm$.

Analogously to \eqref{eq:prod}, define the likelihood ratio perturbation $\delta^{\U}_u: \mathsf{child}(u)\to\R$ by
\begin{equation}
    \delta^{\U}_u(v) \coloneqq \frac{p_{\U}(v|u)}{p_0(v|u)} - 1 = \bra{\psi^u_v}(\rho_{\U}^{\otimes k} - \Id)\ket{\psi^u_v}.
\end{equation}
The ratio between the probabilities of reaching node $u$ when the learner is measuring copies of $\rho_{\U}$ versus copies of $\rhomm$ can thus be expressed as:
\begin{equation}
    L_{\U}(u) \coloneqq \prod^t_{i=1} \left(1 + \delta^{\U}_{u_{i-1}}(u_i)\right),
\end{equation}
where $\brc{r = u_0, u_1,\ldots, u_{t-1}, u_t = u}$ is the root-to-leaf path to $u$. Also define \begin{equation}
    L(u) \coloneqq \E[\U]{L_{\U}(u)}\,.
\end{equation}
Lastly, crucial to the calculations in this section is the following pairwise correlation quantity
\begin{equation}
    \phi^{\U,\V}_u \coloneqq \iprod{\delta^{\U}_u, \delta^{\V}_u}_{p_0(\cdot|u)},
\end{equation}
recalling the notational convention of \eqref{eq:iprod}.

\subsection{Basic Facts}

We note that as the Haar measure on $U(d)$ is a $k$-design for any $k$, so we can use the results of Section~\ref{sec:proof_kdesign} to easily conclude the following:

\begin{lemma}\label{lem:LRbound}
    For $k\le O(\min(d,\sqrt{d}/\epsilon))$ and any node $u$ in $\calT$,
    \begin{equation}
        L(u) \ge \exp\left(-O(2N\epsilon^2k^2/d)\right).
    \end{equation}
\end{lemma}

\begin{proof}
    This follows immediately from applying \eqref{eq:lateruseful} and \eqref{eq:relubound} from the proof of Lemma~\ref{lem:convexity}, together with Lemma~\ref{lem:secondmoment}.
\end{proof}

We record a simple pointwise upper bound on the likelihood ratio perturbations:

\begin{lemma}\label{lem:triv_bound}
    If $\epsilon \le 1/3k$, then for any $v$ in the support of $p_0(\cdot|u)$ and any $\U\in U(d)$, $|\delta^{\U}_u(v)| \le 1$.
\end{lemma}

\begin{proof}
    For any unit vector $\ket{\psi}$, $0< \bra{\psi}\rho^{\otimes k}_{\U}\ket{\psi} \le (1 + \epsilon)^k < 2$ while $\bra{\psi}\Id\ket{\psi} = 1$, from which the claim follows.
\end{proof}

Finally, we will need the following elementary estimate in subsequent sections:

\begin{fact}[Integration by parts, see e.g. Fact C.2 in \cite{bubeck2020entanglement}]\label{fact:stieltjes}
	Let $a,b\in\R$. Let $Z$ be a nonnegative random variable satisfying $Z\le b$ and such that for all $x\ge a$, $\Pr{Z > x} \le \tau(x)$. Let $f: [0,b]\to\R_{\ge 0}$ be nondecreasing and differentiable. Then \begin{equation}
		\E{f(Z)} \le f(a)(1 + \tau(a)) + \int^b_a \tau(x) f'(x) dx\,.
	\end{equation}
\end{fact}

\subsection{Tail Bounds}

In this subsection we prove some tail bounds for $\phi^{\U,\V}_u$ and related quantities. 

First, fix an arbitrary node $u$ in $\calT$. Define the functions $F^{\V}_u(\U)\coloneqq \phi^{\U,\V}_u$ and $G_u(\U)\coloneqq \E[v\sim p_0(\cdot|u)]*{\delta^{\U}_u(v)^2}^{1/2}$. In the rest of this subsection, we will drop subscripts referring to $u$ and denote $p_0(\cdot|u)$ by $p_0$ whenever convenient. Denote the POVM at node $u$ by $\ketbra{\psi_v}$.

As our primary tool for establishing concentration is Theorem~\ref{thm:levy}, we would like to show that the functions $F^{\V}_u$ and $G_u$ are Lipschitz. To this end, we begin by proving the following:

\begin{lemma}\label{lem:second_lip}
    $\E[v\sim p_0]*{(\delta^{\U}(v) - \delta^{\U'}(v))^2}^{1/2} \lesssim \frac{2\epsilon k}{\sqrt{d}}\cdot (1 + \epsilon^2)^{(k-1)/2}\cdot \norm{\U -\U'}_{HS}$.
\end{lemma}

\begin{proof}
    Write $\vec{A} \coloneqq \frac{1}{\epsilon}\left((\Id + \epsilon \U\Z\U^{\dagger})^{\otimes k} - (\Id + \epsilon \U'\Z{\U'}^{\dagger})^{\otimes k}\right)$ so that
    \begin{equation}
         \delta^{\U}(v) - \delta^{\U'}(v) = \epsilon \bra{\psi_v}\vec{A}\ket{\psi_v}.
    \end{equation}
    Take the eigendecomposition $A = \W\Sig\W^{\dagger}$. Then
    \begin{align}
        \bra{\psi_v}\vec{A}\ket{\psi_v}^2 &= \Tr\left(\Sig\cdot \vec{W}^{\dagger}\ketbra{\psi_v}\vec{W}\right)^2 \\
        &= \left(\sum_i \Sig_{ii} \left(\vec{W}^{\dagger}\ketbra{\psi_v}\vec{W}\right)_{ii}\right)^2 \\
        &\le \sum_i \Sig^2_{ii} \left(\vec{W}^{\dagger}\ketbra{\psi_v}\vec{W}\right)_{ii},
    \end{align}
    where the third step follows by Jensen's and the fact that the diagonal entries of $\vec{W}^{\dagger}\ketbra{\psi_v}\vec{W}$ define a probability distribution.
    As $\E[v\sim p_0]*{\ketbra{\psi_v}} = \frac{1}{d^k}\Id$, we conclude that 
    \begin{equation}
        \E[v\sim p_0]*{(\delta^{\U}(v) - \delta^{\U'}(v))^2} = \frac{\epsilon^2}{d}\sum_i \Sig^2_{ii} = \frac{\epsilon^2}{d^k}\norm{\vec{A}}^2_{HS}
    \end{equation}
    Finally, note that
    \begin{equation}
        \norm{\vec{A}}_{HS} \le \sum^k_{j=1} \norm*{(\Id + \epsilon\U\Z\U^{\dagger})^{\otimes k - j}\otimes (\Id + \epsilon\U'\Z{\U'}^{\dagger})^{\otimes j - 1}\otimes (\U^{\dagger}\Z\U - \U'^{\dagger}\Z\U')}_{HS},
    \end{equation}
    so because $\norm{\Id + \epsilon\U\Z\U^{\dagger}}_{HS} \le (d + \epsilon^2 d)^{1/2}$ and
    \begin{equation}
        \norm{\U^{\dagger}\Z\U - \U'^{\dagger}\Z\U'}_{HS} = \norm{\U\Z(\U - \U') + (\U - \U')\Z\U'}_{HS} \le 2\norm{\U - \U'}_{HS},
    \end{equation}
    we get
    \begin{equation}
        \frac{\epsilon}{d^{k/2}}\cdot\norm{\vec{A}}_{HS} \le 2\epsilon k\cdot \frac{(d + \epsilon^2 d)^{(k-1)/2}}{d^{k/2}}\cdot \norm{\U - \U'}_{HS} = \frac{2\epsilon k}{\sqrt{d}}\cdot (1 + \epsilon^2)^{(k-1)/2}\cdot \norm{\U -\U'}_{HS}.
    \end{equation}
    as claimed.
\end{proof}

\begin{corollary}\label{cor:FGlip}
    The functions $F^{\V}_u$ and $G_u$ are $\frac{2\epsilon k}{\sqrt{d}}\cdot (1 + \epsilon^2)^{(k-1)/2}\cdot G(\V)$- and $\frac{2\epsilon k}{\sqrt{d}}\cdot (1 + \epsilon^2)^{(k-1)/2}$-Lipschitz respectively.
\end{corollary}

\begin{proof}
    We have
    \begin{align}
        F^{\V}(\U) - F^{\V}(\U') &= \E[v\sim p_0]{(\delta^{\U}(v) - \delta^{\U'}(v))\delta^{\V}(v)} \\
        &\le  \E[v\sim p_0]*{(\delta^{\U}(v) - \delta^{\U'}(v))^2}^{1/2}\cdot \E[v\sim p_0]*{\delta^{\U}(v)^2}^{1/2} \\
        &= \E[v\sim p_0]*{(\delta^{\U}(v) - \delta^{\U'}(v))^2}^{1/2}\cdot G(\U) \label{eq:decomp}
    \end{align}
    By triangle inequality for $L_2$ norm with respect to the measure $p_0$,
    \begin{equation}
        G(\U)  - G(\U') = \E[v]{\delta^{\U}(v)^2}^{1/2} - \E[v]{\delta^{\U'}(v)^2}^{1/2} \le \E[v]{(\delta^{\U}(v) - \delta^{\U'}(v))^2}^{1/2}.
    \end{equation} Lipschitzness of $F^{\V}$ and $G$ follows from Lemma~\ref{lem:second_lip}.
\end{proof}

Next, we bound the expectation of $G(\U)$. The proof follows immediately from the calculations in the proof of Lemma~\ref{lem:secondmoment}:

\begin{lemma}\label{lem:GV2}
    If $k = O(\min(d,\sqrt{d}/\epsilon))$, then $\E[\U]{G_u(\U)} \le O(\epsilon k/\sqrt{d})$.
\end{lemma}

\begin{proof}
    By Jensen's, $\E[\U]{G(\U)} \le \E[\U]{\E[v\sim p_0]{\delta^{\U}(v)^2}}^{1/2}$. We will bound the latter. Similar to \eqref{eq:summands_odd}, we can write
    \begin{equation}
        \E[\U]{\E[v\sim p_0]{\delta^{\U}(v)^2}} = \sum_{S,S'\neq\emptyset} \epsilon^{|S| + |S'|} \E*{\Tr\left((\U^{\dagger})^{\otimes|S|}\rho_S \U^{\otimes|S|}\Z^{\otimes|S|}\right)\Tr\left((\U^{\dagger})^{\otimes|{S'}|}\rho_{S'} \U^{\otimes|{S'}|}\Z^{\otimes|{S'}|}\right)}.
    \end{equation}
    Recall from the proof of Lemma~\ref{lem:secondmoment} that the $(S,S')$-th summand above is upper bounded by $(\epsilon^2 m/d)^{m/2}$; it is also easy to see that when $|S| + |S'|$ is odd, the summand vanishes. We can thus upper bound $\E[\U]{\E[v\sim p_0]{\delta^{\U}(v)^2}}$ by
    \begin{align}
        \sum_{\substack{S,S'\neq\emptyset \\ |S| + |S'| \ \text{even}}} O((|S| + |S'|)/d)^{|S|/2 + |S'|/2}
        &= \sum_{2\le m\le k \ \text{even}} O(\epsilon^2m/d)^{m/2}\cdot \sum_{0 < i < m} \binom{k}{i}\binom{k}{m-i} \\
        &\le \sum_{2\le m\le k \ \text{even}} O(\epsilon^2 m/d)^{m/2}\cdot \binom{2k}{m} \le O(\epsilon^2 k^2/d),
    \end{align}
    where in the last step we used the assumption that $k\le c\sqrt{d}/\epsilon$ in the same way as in the proof of Lemma~\ref{lem:secondmoment}.
\end{proof}

From Theorem~\ref{thm:levy}, Corollary~\ref{cor:FGlip}, and Lemma~\ref{lem:GV2}, we conclude that $G_u$ satisfies the following tail bound.

\begin{lemma}\label{lem:Gtail}
    There is an absolute constant $c > 0$ such that for any $s > 0$, \begin{equation}
        \Pr[\U]*{G_u(\U) > \frac{c \epsilon k}{\sqrt{d}} + s} \le \exp\left(-s^2/\sigma^2\right)
    \end{equation}
    for $\sigma = \Theta\left(\frac{\epsilon k}{d}\cdot (1 + \epsilon^2)^{(k-1)/2}\right)$.
\end{lemma}



Finally, we show how to combine all of these ingredients to show a tail bound for the pairwise correlation $\phi^{\U,\V}_u$.

\begin{lemma}\label{lem:phitail}
    For any $s > 0$, \begin{equation}
        \Pr*{\abs*{\phi^{\U,\V} - \E*{\phi^{\U,\V}}} > s} \le \exp\left(-\Omega\left(\min\brc*{\frac{d^2 s}{\epsilon^2 k^2(1 + \epsilon^2)^{k-1}}, \frac{d^3s^2}{\epsilon^4 k^4 (1 + \epsilon^2)^{k-1}}}\right)\right).
    \end{equation}
\end{lemma}

\begin{proof}
    By Corollary~\ref{cor:FGlip}, $F^{\V}$ is $\sigma\sqrt{d}\cdot G(\V)$-Lipschitz for $\sigma = \Theta(\frac{\epsilon k}{d}\cdot (1 + \epsilon^2)^{(k-1)/2})$. Define $\mu^{\V}\coloneqq \E[\U]{F^{\V}(\U)}$. By Theorem~\ref{thm:levy}, there is an absolute constant $C > 0$ such that
    \begin{equation}
        \Pr[\U]*{\abs*{F^{\V}(\U) - \mu^{\V}} > s} \le \exp\left(-\frac{C s^2}{\sigma^2G(\V)}\right). \label{eq:fvtail}
    \end{equation}
    We would like to further integrate over $\V$ to get a tail bound for $\phi^{\U,\V}$. To this end, apply Fact~\ref{fact:stieltjes} to the random variable $Z\coloneqq G(\V)$ and $f(Z)\coloneqq \exp\left(-\frac{C s^2}{\sigma^2 Z^2}\right)$. By Lemma~\ref{lem:Gtail}, for all $x\ge a$ for $a\coloneqq \frac{2c\epsilon k}{\sqrt{d}}$, we have $\Pr{Z > x} \le \tau(x)$ for \begin{equation}
        \tau(x) \coloneqq \exp\left(-\frac{C'}{\sigma^2}\left(x - \frac{c\epsilon k}{\sqrt{d}}\right)^2\right)
    \end{equation}
    for absolute constant $C' > 0$.
    
    Fact~\ref{fact:stieltjes} therefore ensures that
    \begin{equation}
        \E{f(Z)} \le 2\exp\left(-\frac{C d s^2}{4c^2\epsilon^2 k^2 \sigma^2}\right) + \int^1_{2c\epsilon k/\sqrt{d}} \frac{2C s^2}{\sigma^2 x^3}\exp\left(-\frac{1}{\sigma^2}\brk*{\frac{C s^2}{x^2} + C'\left(x - \frac{c\epsilon k}{\sqrt{d}}\right)^2}\right) dx. \label{eq:expf}
    \end{equation}
    By AM-GM, for $x \ge 2c\epsilon k/\sqrt{d}$, 
    \begin{equation}
        \frac{C s^2}{x^2} + C'\left(x - \frac{c\epsilon k}{\sqrt{d}}\right)^2 \ge 2(CC')^{1/2} s\left(1 - \frac{c\epsilon k/\sqrt{d}}{x}\right) \ge (CC')^{1/2} s,
    \end{equation}
    so we conclude from \eqref{eq:fvtail} and \eqref{eq:expf} that
    \begin{align}
        \Pr{\abs*{\phi^{\U,\V} - \E{\phi^{\U,\V}}} > s} &\le 2\exp\left(-\frac{C d s^2}{4c^2\epsilon^2 k^2 \sigma^2}\right) + \frac{Cd s^2}{c^2\epsilon^2 k^2\sigma^2}\cdot \exp\left(-\Omega(s/\sigma^2)\right) \\
        &\le \exp\left(-\Omega\left(\min\brc*{\frac{s}{\sigma^2}, \frac{ds^2}{\epsilon^2 k^2\sigma^2}}\right)\right).
    \end{align}
    The lemma follows upon substituting our choice of $\sigma$.
\end{proof}

It remains to compute $\E*{\phi^{\U,\V}}$. Note that when $k = 1$, this expectation is zero, but this is not true for general $k$.

\begin{lemma}\label{lem:phiexp}
    If $k = O(\sqrt{d})$, then $\E{\phi^{\U,\V}} \le O(k^4\epsilon^4/d^2)$
\end{lemma}

\begin{proof}
    We can write $\phi^{\U,\V}$ explicitly as
    \begin{equation}
        \phi^{\U,\V} = \sum_{S,S'\subseteq[k]: S,S'\neq\emptyset} \epsilon^{|S|+|S'|}\E[v\sim p_0]*{\bra{\psi_v}(\U\Z\U^{\dagger}\otimes_S \Id)\ket{\psi_v}\bra{\psi_v}(\V\Z\V^{\dagger}\otimes_{S'} \Id)\ket{\psi_v}} \label{eq:phi_explicit}
    \end{equation}

    Let $\rho^v_S$ denote the the state $\ketbra{\psi_v}$ with the components outside of $S$ traced out. Note that \begin{equation}
        \E[v\sim p_0]{\rho^v_S} = \frac{1}{d^{|S|}}\Id^{\otimes |S|}.
    \end{equation}
    Consider the $(S,S')$-th expectation in \eqref{eq:phi_explicit}. Its expectation with respect to $\U,\V$ is given by
    \begin{equation}
        \sum_{\substack{\pi,\tau\in\calS_{|S|} \\ \pi', \tau'\in\calS_{|S'|}}} \Wg(\pi\tau^{-1},d) \Wg(\pi'{\tau'}^{-1},d)  \Tr(P_{\tau} \Z^{\otimes |S|}) \Tr(P_{\tau'} \Z^{\otimes |S'|}) \E[v\sim p_0]*{\Tr(P_{\pi} \rho^v_S) \Tr(P_{\pi'} \rho^v_{S'})}, \label{eq:wgwg}
    \end{equation}
    We can na\"{i}vely upper bound the last quantity in \eqref{eq:wgwg}, namely the expectation over $v$, by 1.
    We can now decouple the sum in \eqref{eq:wgwg} into a product of a sum over $\pi,\tau$ and a sum over $\pi',\tau'$. Before we do this however, let us recall basic estimates for the other terms in \eqref{eq:wgwg}. First, recall that $\Tr(P_{\tau}\Z^{\otimes|S|}) = 0$ if $\tau$ has an odd cycle and otherwise equals $d^{\#\tau}$, and similarly for $\Tr(P_{\tau'}\Z^{\otimes|S'|})$. As for the Weingarten terms, by Lemma~\ref{lem:montanaro} we have $|\Wg(\pi\tau^{-1},d)| \le O(d^{\#\pi\tau^{-1} - 2|S|})$ as $|S| \le k \le d^{2/3}$, and similarly for $|\Wg(\pi'{\tau'}^{-1},d)|$.
    
    Putting everything together, we can upper bound \eqref{eq:wgwg} by 
    \begin{equation}
        \frac{1}{\Omega(d)^{2|S| + 2|S'|}}\left(\sum_{\pi,\tau\in\calS_{|S|}: \tau \, \text{even}} O(d)^{\#\pi\tau^{-1} +\#\tau}\right)\left(\sum_{\pi',\tau'\in\calS_{|S'|}: \tau \,\text{even}} O(d)^{\#\pi'{\tau'}^{-1} +\#\tau'}\right) \label{eq:wgwg2}
    \end{equation} when $|S|$ and $|S'|$ are even, and by zero otherwise. Here ``$\tau$ even'' means that $\tau$ consists only of even cycles.
    
    For $|S|$ even,
    \begin{align}
        \sum_{\pi,\tau\in\calS_{|S|}: \tau \, \text{even}} d^{\#\pi\tau^{-1} +\#\tau} &= \left(\sum_{\pi\in\calS_{|S|}} d^{\#\pi}\right)\left(\sum_{\tau\in\calS_{|S|} \ \text{even}} d^{\#\tau}\right) \\
        &= (|S| - 1)!! \cdot  d(d+2)\cdots(d+|S|-2)\cdot d(d+1)\cdots (d + |S|-1), \\
        &\le O(d)^{3|S|/2}\cdot |S|^{|S|/2}
    \end{align}
    when $|S| \le k \le O(d)$, where we used Facts~\ref{fact:sumcycles} and \ref{fact:evencycles} in the second step.
    
    Substituting this into \eqref{eq:wgwg2} gives an upper bound of
    \begin{equation}
        \frac{1}{\Omega(d)^{|S|/2 + |S'|/2}}\cdot |S|^{|S|/2} \cdot |S'|^{|S'|/2}\,.
    \end{equation}
    We conclude that
    \begin{equation}
        \E{\phi^{\U,\V}} = \left(\sum_{S\subseteq[k]: S\neq\emptyset,\, |S| \, \text{even}} O(\epsilon^2 |S|/d)^{|S|/2}\right)^2.
    \end{equation}
    Finally, note by a similar calculation as in the end of the proof of Lemma~\ref{lem:GV2},
    \begin{align}
        \sum_{S\subseteq[k]: S\neq\emptyset, |S| \ \text{even}} O(\epsilon^2 |S|/d)^{|S|/2} &= \sum^k_{j=2 \ \text{even}} (\epsilon^2 j/d)^{j/2} \binom{k}{j} \\
        &\le \sum^k_{j=2 \ \text{even}} O\left(\epsilon k/\sqrt{dj}\right)^j,
    \end{align}
    and so because $k \le O(\sqrt{d}) \le O(\sqrt{d}/\epsilon)$ the leading term dominates, and we get the claimed bound on $\E{\phi^{\U,\V}}$.
\end{proof}

\subsection{Consequences of Tail Bounds}

We now record some moment bounds that follow from the tail bounds of the previous section. The main estimate we wish to prove in this subsection is the following upper bound on $\E[u\sim p^t_0,\U,\V]*{L_{\U}(u)^2 L_{\V}(u)^2}$:

\begin{lemma}\label{lem:oneplusZimpliesPsi}
    For any $t\in\mathbb{N}$, $\E[u\sim p^t_0,\U,\V]*{L_{\U}(u)^2 L_{\V}(u)^2} \le \exp(O(1 + \epsilon^2 t k^2/d))$.
\end{lemma}

\noindent To prove this, it will be convenient to define \begin{equation}
	K^{\U,\V}_u \coloneqq \E[v\sim p_0(\cdot|u)]*{\left(\delta^{\U}_u(v) + \delta^{\V}_u(v)\right)^2}
\end{equation} and first show the following:

\begin{lemma}\label{lem:tail_imply_oneplus}
	Let $\sigma = \Theta\left(\frac{\epsilon k}{d}\cdot (1 + \epsilon^2)^{(k-1)/2}\right)$ be the parameter defined in Lemma~\ref{lem:Gtail}. For any node $u$ in the tree $\calT$, any $\gamma>0$, and any $t \le 1/(16\gamma\sigma^2)$, $\E[\U,\V]*{\left(1 + \gamma K^{\U,\V}_u\right)^t} \le \exp(O(1 + \gamma\epsilon^2 t k^2/d))$.
\end{lemma}

\begin{proof}
	Note that $(K^{\U,\V}_{u})^{1/2} \le G_u(\U) + G_u(\V)$, so by Lemma~\ref{lem:Gtail} and a union bound, \begin{equation}
	    \Pr*{K^{\U,\V}_u > \frac{8c^2\epsilon^2 k^2}{d} + 8s^2} \le \Pr*{K^{\U,\V}_u > \left(\frac{2c\epsilon k}{\sqrt{d}} + 2s\right)^2} \le 2\exp\left(-s^2/\sigma^2\right)    
	\end{equation} 
	for any $s > 0$. 

	We can apply Fact~\ref{fact:stieltjes} to the random variable $Z\coloneqq K^{\U,\V}_u$ and the function $f(Z) \coloneqq (1 + \gamma Z)^t$ to conclude that for $\mu \coloneqq \frac{8c^2\epsilon^2 k^2}{d}$, \begin{align}
		\E[\U,\V]*{\left(1 + \gamma\cdot K^{\U,\V}_u\right)^t} &\le 2(1 + \mu)^t + \int^{\infty}_{\mu} \gamma t(1 + \gamma x)^{t-1} \cdot e^{-(x - \mu)/8\sigma^2} dx  \\
		&\le 2(1 + \gamma\mu)^t + \gamma t e^{\mu/8\sigma^2} \int^{\infty}_0 e^{-x(1/8\sigma^2 - \gamma(t-1))} dx \\
		&\le 2(1 + \gamma\mu)^t + \frac{\gamma t e^{\mu/8\sigma^2}}{1/8\sigma^2 - \gamma(t-1)} \\
		&\le 2(1 + \gamma\mu)^t + e^{\mu/8\sigma^2} \le \exp(O(1+\gamma\mu t)) + O(1),
	\end{align}
	where in the third and fourth steps we used that $t \le 1/(16\gamma\sigma^2)$.
\end{proof}

We can now prove Lemma~\ref{lem:oneplusZimpliesPsi}:

\begin{proof}[Proof of Lemma~\ref{lem:oneplusZimpliesPsi}]
	By Lemma~\ref{lem:triv_bound}, we know that for any $a,b \ge 2$, \begin{equation}
		\E[v\sim p_0(\cdot|u)]*{\delta^{\U}_u(v)^a\cdot \delta^{\V}_u(v)^b} \le \E[v\sim p_0(\cdot|u)]{\delta^{\U}_u(v)^2},
	\end{equation}
	so we conclude that \begin{equation}
		\E[v]*{(1 + \delta^{\U}_{u}(v))^c(1 + \delta^{\V}_{u}(v))^c} \le 1 + O_c\left(\E[v]{\delta^{\U}_{u}(v)^2}\right) + O_c\left(\E[v]{\delta^{\V}_{u}(v)^2}\right) + O_c\left(\phi^{\U,\V}_{u}\right) \le 1 + C(c)\cdot K^{\U,\V}_{u}\label{eq:gc}
	\end{equation} for some absolute constant $C(c)>0$, where the last step follows by AM-GM. For $\alpha_i \coloneqq 2\cdot\left(\frac{t-1}{t-2}\right)^i$, we have that 
	\begin{align}
		\MoveEqLeft \E[u\sim p^t_0,\U,\V]*{\left(L_{\U}(u) L_{\V}(u)\right)^{\alpha_i}} \\
		&\le \E[u'\sim p^{t-1}_0,\U,\V]*{\left(L_{\U}(u') L_{\V}(u')\right)^{\alpha_i} \cdot \left(1 + C(\alpha_i)\cdot K^{\U,\V}_{u'}\right)} \label{eq:apply_gc} \\
		&\le \E[u'\sim p^{t-1}_0,\U,\V]*{\left(L_{\U}(u') L_{\V}(u')\right)^{\alpha_i (t-1)/(t-2)}}^{(t-2)/(t-1)}\cdot \E[u'\sim p^{t-1}_0,\U,\V]*{\left(1 + C(\alpha_i)\cdot K^{\U,\V}_{u'}\right)^{t-1}}^{1/(t-1)}\label{eq:holder} \\
		&\le \E[u'\sim p^{t-1}_0,\U,\V]*{\left(L_{\U}(u') L_{\V}(u')\right)^{\alpha_{i+1} (t-1)/(t-2)}}\cdot \E[u'\sim p^{t-1}_0,\U,\V]*{\left(1 + C(\alpha_i)\cdot K^{\U,\V}_{u'}\right)^{t-1}}^{1/(t-1)}.
	\end{align} where \eqref{eq:apply_gc} follows by \eqref{eq:gc}, and \eqref{eq:holder} follows by H\"{o}lder's inequality. Unrolling this recurrence, we conclude that
	\begin{align}
		\E[u\sim p^t_0,\U,\V]*{\left(L_{\U}(u)L_{\V}(u)\right)^2} &\le \prod^{t-1}_{i=1}\E[u'\sim p^i_0,\U,\V]*{\left(1 + C(\alpha_{t-1-i})\cdot K^{\U,\V}_{u'}\right)^{t-1}}^{1/(t-1)}\label{eq:unroll} \\
		&\le \prod^{t-1}_{i=1}\E[u'\sim p^i_0,\U,\V]*{\left(1 + C(2e)\cdot K^{\U,\V}_{u'}\right)^{t-1}}^{1/(t-1)},\label{eq:2e} \\
		&\le \sup_{u\in\calT}\E[\U,\V]*{\left(1 + O(K^{\U,\V}_u)\right)^{t-1}}
	\end{align} where \eqref{eq:2e} follows by the fact that for $1\le i \le t - 1$, $\alpha_{t-1-i} \le 2\left(1 + \frac{1}{t-2}\right)^{t-2} \le 2e$, and the supremum in the last step is over all POVMs $\calM$. The lemma then follows from Lemma~\ref{lem:tail_imply_oneplus}.
\end{proof}

\subsection{Putting Everything Together}

We now show how to combine the ingredients from the previous subsections to conclude the proof of Theorem~\ref{thm:mixed}. We will need the following consequence of the chain rule from \cite{bubeck2020entanglement}:

\begin{lemma}[Lemma 6.1 from \cite{bubeck2020entanglement}]\label{lem:chain}
	\begin{equation}
		\text{\rm KL}{\E{p_{\rho_{\U}}}}{p_{\rhomm}} \le \sum^{N-1}_{t=0}\E[u\sim p^t_0]*{\frac{1}{L(u)}\E[\U,\V]*{L_{\U}(u)L_{\V}(u)\cdot \phi^{\U,\V}_{u}}},
	\end{equation}
\end{lemma}

\begin{proof}[Proof of Theorem~\ref{thm:mixed}]
	Given a fixed node $u$ of the tree $\calT$ and Haar-random $\U,\V$, let $\calE^{\U,\V}_u$ denote the event that $\abs*{\phi^{\U,\V}_{u}} > \tau$ for \begin{equation}
	    \tau\coloneqq \frac{\epsilon^4k^4}{d^2} + \epsilon^2 k^2(1 + \epsilon^2)^{k-1}\cdot d^{-3/2}\sqrt{d^{1/3} + \log^2 N} \le O\left(\epsilon^2 k^2(1 + \epsilon^2)^{k-1}\cdot \max(d^{-4/3}, d^{-3/2}\log N)\right)
	\end{equation} By Lemma~\ref{lem:phitail}, for any node $u$ we have
	\begin{equation}
	    \Pr{\calE^{\U,\V}_u}\le \exp\left(-\Omega\left(\min\brc*{\sqrt{d^{4/3} + d\log^2 N}, (d^{1/3} + \log^2 N) (1 + \epsilon^2)^{k-1}}\right)\right) \le \exp\left(-\Omega(d^{1/3} + \log N)\right) \label{eq:Ebound}
	\end{equation}
    We have that \begin{align}
		\E[\U,\V]*{L_{\U}(u)L_{\V}(u)\cdot \phi^{\U,\V}_u} &= \E[\U,\V]*{L_{\U}(u)L_{\V}(u)\cdot \phi^{\U,\V}_u \cdot \left(\bone{\calE_{u}^{\U,\V}} + \bone{(\calE_{u}^{\U,\V})^c}\right)} \\
		&\le \E[\U,\V]*{L_{\U}(u)L_{\V}(u)\cdot \bone{\calE_{u}^{\U,\V}}} + \tau \E[\U,\V]*{L_{\U}(u)L_{\V}(u)\cdot \bone{(\calE_{u}^{\U,\V})^c}} \\
		&\le \E[\U,\V]*{L_{\U}(u)L_{\V}(u)\cdot \bone{\calE_{u}^{\U,\V}}} + \tau L(u)^2,
	\end{align} where in the second step we used Lemma~\ref{lem:triv_bound} to conclude that $\phi^{\U,\V}_{u} \le 1$. 
	By Lemma~\ref{lem:LRbound} and the fact that the likelihood ratio integrates to 1, we can thus bound the $t$-th summand in the upper bound of Lemma~\ref{lem:chain} by \begin{equation}
		\E[u\sim p^t_0]*{\frac{1}{L(u)}\E[\U,\V]*{L_{\U}(u)L_{\V}(u)\cdot \phi^{\U,\V}_{u}}} \le \exp\left(-O(2N\epsilon^2k^2/d)\right) \cdot \E[u\sim p^t_0,\U,\V]*{L_{\U}(u)L_{\V}(u)\cdot \bone{\calE_{u}^{\U,\V}}} + \tau. \label{eq:puttogether}
	\end{equation}

	To bound the expectation on the right-hand side of \eqref{eq:puttogether}, we apply Cauchy-Schwarz to get \begin{align}
		\E[u\sim p^t_0,\U,\V]*{L_{\U}(u)^2 L_{\V}(u)^2}^{1/2}\cdot \Pr[u,\U,\V]*{\calE_{u}^{\U,\V}}^{1/2} \le \exp(O(1 + \epsilon^2 Nk^2/d))\cdot \exp(-\Omega(d^{1/3} + \log N)),
	\end{align}
	where the second step follows by Lemma~\ref{lem:oneplusZimpliesPsi} and the bound in \eqref{eq:Ebound}. Note that when $N = O(d^2/(\epsilon^2 k^2))$, this bound is $O(1/n)$, where the constant factor can be made arbitrarily small. Invoking Lemma~\ref{lem:chain}, we conclude that
	\begin{equation}
	    \KL{\E{p_{\rho_{\U}}}}{p_{\rhomm}} \le N\left(O(1/N) + \tau\right) = N\tau + O(1).
	\end{equation}
	So provided that 
	\begin{equation}
	    N\le \min\left(O\left(\frac{d^{4/3}}{\epsilon^2 k^2(1+\epsilon^2)^{k-1}}\right), \widetilde{O}\left(\frac{d^{3/2}}{\epsilon^2 k^2(1+\epsilon^2)^{k-1}}\right)\right),
	\end{equation}
	the KL divergence between $\E{p_{\rho_{\U}}}$ and $p_{\rhomm}$ is a small constant. The theorem follows from Pinsker's and Lemma~\ref{lem:lecam}.
\end{proof}

\subsection{Nonadaptive Lower Bound}

In this section, we show that in the weaker \emph{nonadaptive} setting where one prepares a sequence of $k$-entangled POVMs in advance and performs them in succession on batches of $k$ replicas, we can show an optimal lower bound on the copy complexity needed for mixedness testing:

\begin{theorem}\label{thm:nonadaptive}
    Suppose $k\in\mathbb{N}$ satisfies $k\le O(d^{1/2})$ and $k \ll \log(d)/\epsilon^2$. Any nonadaptive protocol which can make entangled measurements of at most $k$ replicas of $\rho$ simultaneously must use at least $\Omega(d^{\frac{3}{2} - \eta}/\epsilon^2)$
    copies of $\rho$ overall, for any constant $\eta > 0$, to solve the mixedness testing task with constant probability. 
\end{theorem}

This lower bound is optimal in the sense that there is a matching algorithm making nonadaptive \emph{1-entangled} measurements that only needs $O(d^{3/2}/\epsilon^2)$ copies of $\rho$. We can thus interpret Theorem~\ref{thm:nonadaptive} as saying that in the nonadaptive setting, the optimal thing to do if one can only measure $k \ll \log(d)/\epsilon^2$ replicas at a time is merely to measure a single replica at a time!

We now proceed to the proof of Theorem~\ref{thm:nonadaptive}. Note that if the POVMs are chosen nonadaptively in advance, then if $\rho = \rho_{\U}(\epsilon)$ for a fixed $\U$, the induced distribution over the transcript of measurement outcomes is simply a product distribution, as the choice of POVM at any given step does not depend on the previous measurement outcomes. The same holds if $\rho = \rhomm$. This significantly simplifies the expression for the chi-squared divergence between the distribution over transcripts under $\rho = \rhomm$ and the distribution under $\rho = \rho_{\U}(\epsilon)$ for Haar-random $\U$.

Indeed, we can use the following lemma giving an expression for chi-squared divergence between product mixtures:

\begin{lemma}[Lemma 2.8 from \cite{bubeck2020entanglement}]\label{lem:ingster}
	If the learning algorithm is nonadaptive, then \begin{equation}
		\chi^2\Big(\E[\U]{p_{\rho_{\U}}} \,\Big\|\, p_{\rhomm}\Big) \le \max_{u\in\calT} \E[\U,\V\sim\calD]*{\left(1 + \phi^{\U,\V}_u\right)^N} - 1\label{eq:ingster}
	\end{equation}
\end{lemma}

In other words, to show a lower bound against nonadaptive learning algorithms, it suffices to show bounds on the moments of $\phi^{\U,\V}$ as a random variable in $\U,\V$. These immediately follow from the tail bound in Lemma~\ref{lem:phitail}. We can now use this to prove Theorem~\ref{thm:nonadaptive}:

\begin{proof}[Proof of Theorem~\ref{thm:nonadaptive}]
	Define $\overline{\phi}^{\U,\V} \coloneqq \phi^{\U,\V} - \E[\U,\V]{\phi^{\U,\V}}$. We have \begin{align}
		\E*{\left(\overline{\phi}^{\U,\V}\right)^t} &= \int^{\infty}_0 \Pr*{\abs*{\overline{\phi}^{\U,\V}}> s^{1/t}} ds \\
		&\le \int^{\infty}_0 \exp\left(-\Omega\left(\frac{d^3s^{2/t}}{\epsilon^4k^4(1+\epsilon^2)^{k-1}}\right)\right)  ds + \int^{\infty}_0 \exp\left(-\Omega\left(\frac{d^2s^{1/t}}{\epsilon^2 k^2(1+\epsilon^2)^{k-1}}\right)\right) ds \\
		&= \Gamma(1 + t/2) \cdot O(\epsilon^2(1 + \epsilon^2)^{(k-1)/2}/d^{3/2})^t + \Gamma(1+t)\cdot O(\epsilon^2 k^2(1 + \epsilon^2)^{k-1}/d^2)^t,
	\end{align}
	so
	\begin{equation}
	    \norm{\phi^{\U,\V} - \E{\phi^{\U,\V}}}_t \le O\left(\sqrt{t}\cdot \epsilon^2 k^2(1 + \epsilon^2)^{(k-1)/2}/d^{3/2} + t\cdot \epsilon^2k^2(1+\epsilon^2)^{k-1}/d^2\right).
	\end{equation}
	By triangle inequality for $L_t$ norms of random variables and Lemma~\ref{lem:phiexp},
	\begin{equation}
	    \norm{\phi^{\U,\V}} \le O\left(k^4\epsilon^4/d^2 + \sqrt{t}\cdot \epsilon^2 k^2 (1 + \epsilon^2)^{(k-1)/2}/d^{3/2} + t\cdot \epsilon^2k^2(1+\epsilon^2)^{k-1}/d^2\right)
	\end{equation}
    For any $u\in\calT$, we can expand
    \begin{align}
		\E[\U,\V]*{\left(1 + \phi^{\U,\V}_u\right)^N} - 1 &= \sum_{2\le t\le N \ \text{even}} \binom{N}{t} \E*{\left(\phi^{\U,\V}_{u}\right)^t} \\
		&\le \sum_{2\le t\le N \ \text{even}} O\left(N\left(t^{-1} k^4\epsilon^4/d^2 + t^{-1/2}\cdot \epsilon^2k^2(1 + \epsilon^2)^{(k-1)/2}/d^{3/2} +  \epsilon^2k^2(1+\epsilon^2)^{k-1}/d^2\right)\right)^t.
	\end{align} 
	So when $N \le \frac{cd^{3/2}}{\epsilon^2 k^2(1+\epsilon^2)^{k-1}}$ for sufficiently small constant $c > 0$, then this quantity is a small constant. The lemma then follows from Lemma~\ref{lem:ingster} and Pinsker's inequality.
\end{proof}

\vspace{-0.5cm}

\bibliographystyle{abbrv}
\bibliography{refs}

\end{document}